\documentclass[a4paper, 12pt]{article}

\usepackage[top = 2.5cm, bottom = 2.5cm, left = 2.4cm, right = 2.4cm]{geometry}
\usepackage{amsfonts, amsmath, amsthm, amssymb, mathtools, cases, bbm}
\usepackage{lmodern, indentfirst, setspace, changepage, fancyhdr}
\usepackage{tikz, pgfplots, float, subcaption}
\usepackage[title]{appendix}
\usepackage[english]{babel}
\usepackage{cite}

\usepackage{hyperref}
\hypersetup{
	colorlinks = true,
	linkcolor = blue,
	citecolor = blue,
}

\newtheorem{theorem}{Theorem}
\newtheorem{proposition}{Proposition}
\newtheorem{corollary}{Corollary}
\newtheorem{lemma}{Lemma}
\theoremstyle{definition}
\newtheorem{definition}{Definition}

\newtheorem{example}{Example}
\newtheorem{remark}{Remark}

\newcommand{\ind}[1]{\mathbbm{1}_{\left\{#1\right\}}}
\newcommand{\norm}[1]{\left|\left|#1\right|\right|}
\newcommand{\floor}[1]{\left\lfloor#1\right\rfloor}

\newcommand{\ceil}[1]{\left\lceil#1\right\rceil}
\newcommand{\map}[3]{#1 : #2 \longrightarrow #3}
\newcommand{\set}[2]{\left\{#1 : #2\right\}}

\newcommand{\flambda}{{\floor{\lambda}}}
\newcommand{\clambda}{{\ceil{\lambda}}}

\newcommand{\boldf}{\boldsymbol{f}}
\newcommand{\defeq}{\vcentcolon=}

\newcommand{\ba}{\boldsymbol{a}}
\newcommand{\bb}{\boldsymbol{b}}

\newcommand{\bd}{\boldsymbol{d}}

\newcommand{\bq}{\boldsymbol{q}}

\newcommand{\bs}{\boldsymbol{s}}
\newcommand{\bu}{\boldsymbol{u}}
\newcommand{\bv}{\boldsymbol{v}}
\newcommand{\bw}{\boldsymbol{w}}
\newcommand{\bx}{\boldsymbol{x}}
\newcommand{\by}{\boldsymbol{y}}

\newcommand{\bQ}{\boldsymbol{Q}}

\newcommand{\bV}{\boldsymbol{V}}
\newcommand{\bW}{\boldsymbol{W}}
\newcommand{\bX}{\boldsymbol{X}}
\newcommand{\bY}{\boldsymbol{Y}}
\newcommand{\bZ}{\boldsymbol{Z}}

\newcommand{\prob}{\mathbbm{P}}
\newcommand{\calA}{\mathcal{A}}

\newcommand{\calC}{\mathcal{C}}
\newcommand{\calD}{\mathcal{D}}

\newcommand{\calF}{\mathcal{F}}

\newcommand{\calJ}{\mathcal{J}}
\newcommand{\calK}{\mathcal{K}}
\newcommand{\calL}{\mathcal{L}}

\newcommand{\calN}{\mathcal{N}}

\newcommand{\calR}{\mathcal{R}}

\newcommand{\scdot}{{}\cdot{}}

\newcommand{\eq}{\textrm{eq}}

\newcommand{\N}{\mathbbm{N}}

\newcommand{\R}{\mathbbm{R}}

\newcommand{\e}{\mathrm{e}}
\newcommand{\Zp}{\N}

\usetikzlibrary{automata, arrows, positioning, calc, external, babel}
\pgfplotsset{
	compat = 1.16,
	every axis/.append style = {
		grid style = {dashed, gray, opacity = 0.2},
		label style = {font = \footnotesize},
		tick label style = {font = \footnotesize},  
		width = 1 * \columnwidth,
		height = 0.618 * 1 * \columnwidth
	}
}

\definecolor{britishracinggreen}{rgb}{0.0, 0.26, 0.15}
\definecolor{bostonuniversityred}{rgb}{0.8, 0.0, 0.0}
\definecolor{ceruleanblue}{rgb}{0.16, 0.32, 0.75}
\definecolor{airforceblue}{rgb}{0.36, 0.54, 0.66}
\definecolor{cadmiumgreen}{rgb}{0.0, 0.42, 0.24}
\definecolor{ao(english)}{rgb}{0.0, 0.5, 0.0}
\definecolor{coolblack}{rgb}{0.0, 0.18, 0.39}
\definecolor{alizarin}{rgb}{0.82, 0.1, 0.26}
\definecolor{arsenic}{rgb}{0.23, 0.27, 0.29}
\definecolor{cobalt}{rgb}{0.0, 0.28, 0.67}
\definecolor{amber}{rgb}{1.0, 0.75, 0.0}

\title{Self-Learning Threshold-Based Load Balancing \vspace{\baselineskip}}

\author{
\normalsize{Diego Goldsztajn, Sem C. Borst}\\ \footnotesize{Eindhoven University of Technology, d.e.goldsztajn@tue.nl, s.c.borst@tue.nl} \\
\normalsize{Johan S.H. van Leeuwaarden}\\ \footnotesize{Tilburg University, j.s.h.vanleeuwaarden@uvt.nl} \\
\normalsize{Debankur Mukherjee}\\ \footnotesize{Georgia Institute of Technology, debankur.mukherjee@isye.gatech.edu} \\
\normalsize{Philip A. Whiting}\\ \footnotesize{Macquarie University, philip.whiting@mq.edu.au} \\
}

\date{\vspace{\baselineskip} \normalsize{September 11, 2023}}

\begin{document}

\maketitle

\noindent\rule{\textwidth}{1pt}

\vspace{2\baselineskip}

\onehalfspacing

\begin{adjustwidth}{0.7cm}{0.7cm}
	\begin{center}
		\textbf{Abstract}
	\end{center}
	
	\vspace{0.3\baselineskip}
	
	\noindent We consider a large-scale service system where incoming tasks have to be instantaneously dispatched to one out of many parallel server pools. The user-perceived performance degrades with the number of concurrent tasks and the dispatcher aims at maximizing the overall quality of service by balancing the load through a simple threshold policy. We demonstrate that such a policy is optimal on the fluid and diffusion scales, while only involving a small communication overhead, which is crucial for large-scale deployments. In order to set the threshold optimally, it is important, however, to learn the load of the system, which may be unknown. For that purpose, we design a control rule for tuning the threshold in an online manner. We derive conditions which guarantee that this adaptive threshold settles at the optimal value, along with estimates for the time until this happens. In addition, we provide numerical experiments which support the theoretical results and further indicate that our policy copes effectively with time-varying demand patterns.
	
	\vspace{\baselineskip}
	
	
	\vspace{0.3\baselineskip}
	
	\small{\noindent \textit{Acknowledgment:} the work in this paper is supported by the Netherlands Organisation for Scientific Research (NWO) through Gravitation-grant NETWORKS-024.002.003, by the National Science Foundation (NSF) through Grant No. 2113027 and by the Australian Research Council Discovery Project (DP) through grant DP180103550.} 
\end{adjustwidth}

\section{Introduction}
\label{ch2-sec: introduction}

Consider a service system where incoming tasks have to be immediately routed to one out of $n$ parallel server pools. The service of tasks starts upon arrival and is independent of the number of tasks contending for service at the same server pool. Nevertheless, the portion of shared resources available to individual tasks does depend on the number of contending tasks, and in particular the experienced performance may degrade as the degree of contention rises, creating an incentive to balance the load so as to keep the maximum number of concurrent tasks across server pools as low as possible.

The latter features are characteristic of video streaming applications, such as video conferencing services. In this context, a server pool could correspond to an individual server, instantiated to handle multiple streaming tasks in parallel. The duration of tasks is typically determined by the application and is not significantly affected by the number of instances contending for the finite shared resources of the server (e.g., bandwidth). However, the video and audio quality suffer degradation as these resources get distributed among a growing number of active instances. Effective load balancing policies are thus key to optimizing the overall user experience, but the implementation of these policies must be simple enough as to not introduce significant overheads, particularly in large systems.

Suppose that the processing times of tasks are exponential with unit mean and that tasks arrive as a Poisson process of intensity $n \lambda$. Because all server pools together form an infinite-server system, the total number of tasks in steady state is Poisson with mean $n \lambda$. A natural and simple dispatching strategy is a threshold policy that gives the highest priority to server pools with less than $\floor{\lambda}$ tasks and the second highest priority to server pools with less than $\ceil{\lambda}$ tasks. In fact, we establish that a policy of this kind is optimal on the \emph{fluid scale}: the fraction of server pools with a number of tasks different from $\floor{\lambda}$ or $\ceil{\lambda}$ vanishes over time in a large $n$ regime. We further show that this policy is optimal on the more fine-grained \emph{diffusion scale} and only involves a small implementation overhead.

However, to achieve optimality, the threshold $\floor{\lambda}$ must be learned, since it depends on the demand for service, which may be unknown or even time-varying. For this purpose we introduce a control rule for adjusting the threshold in an online fashion, relying solely on the state information needed to take the dispatching decisions.

Effectively, our policy integrates online resource allocation decisions with demand estimation. While these two attributes are evidently intertwined, the online control actions and longer-term estimation rules are usually decoupled and studied separately in the literature. The former typically assume perfect knowledge of relevant system parameters while the latter tend to focus on statistical estimation of these parameters. In contrast, our policy smoothly blends these two elements and does not rely on an explicit estimate of the load $\lambda$, but yields an implicit indication as a by-product.

\subsection{Main contributions}
\label{ch2-sub: main contributions}

We analyze the threshold policy through fluid and diffusion approximations which are justified by rigorous asymptotic results, and also by means of several numerical experiments. Our main contributions are listed below in more detail.

\begin{itemize}
	\item We show that a threshold dispatching rule is optimal on the fluid and diffusion scales if the threshold is suitably chosen. Moreover, we provide a token-based implementation which involves a low communication overhead, of at most two messages per task, and two bits of memory per server pool at the dispatcher.
	
	\item The optimal threshold value depends on the load $\lambda$, which tends to be unknown or even time-varying in practice. We propose a control rule for adjusting the threshold in an online manner to an unknown load, relying solely on the tokens kept at the dispatcher. In order to analyze this rule, we provide a fluid limit for the joint evolution of the system occupancy and the adaptive threshold.
	
	\item We prove that the threshold settles in finite time in an asymptotic regime, and we provide lower and upper bounds for the equilibrium threshold. These are used to design the control rule for achieving nearly-optimal performance once the threshold has reached an equilibrium. Also, we derive an upper bound for the limit of the time until the threshold settles as $n \to \infty$.
\end{itemize}

The theoretical results are accompanied by simulations, which show that the threshold reaches an equilibrium in systems with a few hundred servers and only after a short time. Furthermore, in the presence of highly variable demand patterns, simulations indicate that the threshold adapts swiftly to demand variations.

\subsection{Related work}
\label{ch2-sub: related literature}

Load balancing problems, similar to the one addressed in the present paper, have received immense attention in the past few decades; \cite{van2018scalable} provides a recent survey. While traditionally the focus in this literature used to be on performance, more recently the implementation overhead has emerged as an equally important issue. This overhead has two sources: the communication burden of exchanging messages between the dispatcher and the servers, and the cost in large-scale deployments of storing and managing state information at the dispatcher, as considered in \cite{gamarnik2018delay}.

While this paper concerns an \emph{infinite-server} setting, the load balancing literature is predominantly focused on single-server models, where performance is generally measured in terms of queue lengths or delays. In that scenario, the Join the Shortest Queue (JSQ) policy minimizes the mean delay for exponentially distributed service times, among all non-anticipating policies; see \cite{ephremides1980simple,winston1977optimality}. However, a naive implementation of this policy involves an excessive communication burden for large systems. So-called power-of-$d$ strategies assign tasks to the shortest among $d$ randomly sampled queues, which involves substantially less communication overhead and yet provides significant improvements in delay performance over purely random routing, even for $d = 2$; see \cite{vvedenskaya1996queueing,mitzenmacher2001power,mukherjee2016universality}. A further alternative are pull-based policies, which were introduced in \cite{badonnel2008dynamic,stolyar2015pull}. These policies reduce the communication burden by maintaining state information at the dispatcher. In particular, the Join the Idle Queue (JIQ) policy studied in \cite{lu2011join,stolyar2015pull} matches the optimality of JSQ in a many-server regime, and involves only one message per task. This is achieved by storing little state information at the dispatcher, in the form of a list of idle queues.

The main differences in the delay performance of the above policies appear in heavy-traffic regimes where the load approaches one. If JSQ is the reference, then JIQ deviates from this benchmark under certain heavy-traffic conditions. This was addressed in \cite{zhou2017designing,zhou2018heavy}, which propose a refinement of JIQ designed to  achieve the same heavy-traffic delay as JSQ at the expense of only a mild increase in the communication overhead; this policy is called Join Below the Threshold (JBT). Despite similarity in name, the problem considered in these papers is fundamentally different from the one addressed in the present paper, because achieving delay optimality in a system of parallel single-server queues does not require to maintain a balanced distribution of the load. In fact, the queue lengths are not balanced if JBT is used.

The JBT policy was considered for systems of heterogeneous limited processor sharing servers with state-dependent service rates in \cite{horvath2019mean,horvath2023mean}. Such servers were studied individually in \cite{gupta2009self}, which analyzes how to set the multi-programming-limit to minimize the mean response time in a way that is robust with respect to the arrival process; this is a scheduling problem where the way in which the service rate changes with the number of tasks sharing the server is a crucial factor. In the context of purely processor sharing servers with finite buffers, \cite{jonckheere2016asymptotics} studies the loss probability of a dispatching policy that is insensitive to the distribution of task sizes. Also, \cite{comte2019dynamic} proposes a token-based insensitive policy for a system with different classes of both tasks and servers, assuming balanced service rates across the server classes.

As mentioned above, the infinite-server setting considered in this paper has received only limited attention in the load balancing literature. While queue lengths and delays are hardly meaningful in this type of scenario, load balancing still plays a crucial role in optimizing different performance measures, and many of the concepts discussed in the single-server context carry over. One relevant performance measure is the loss probability in Erlang-B scenarios; power-of-$d$ properties for these probabilities have been established in \cite{turner1998effect,mukhopadhyay2015mean,mukhopadhyay2015power,xie2015power,karthik2017choosing}. Other relevant measures are Schur-concave utility metrics associated with quality of service as perceived by streaming applications, as considered in \cite{mukherjee2020asymptotic}; these metrics are maximized by balancing the load. As in the single-server setting, JSQ is the optimal policy for evenly distributing tasks among server pools, but it involves a significant implementation burden; see \cite{menich1991optimality,sparaggis1993extremal} for proofs of the optimality of JSQ. It was established in \cite{mukherjee2020asymptotic} that the performance of JSQ can be asymptotically matched by certain power-of-$d$ strategies which reduce the communication overhead significantly, by sampling a suitably chosen number of server pools that depends on the number of tasks and dispatching tasks to the least congested of the sampled server pools.

Just like the algorithms studied in \cite{mukherjee2020asymptotic}, our threshold-based dispatching rule aims at optimizing the overall experienced performance and asymptotically matches the optimality of JSQ on the fluid and diffusion scales. Moreover, this rule involves at most two messages per task and requires that the dispatcher stores only two tokens per server pool; thus, our policy is the counterpart of JIQ in the infinite-server setup. Another token-based algorithm, for an infinite-server \emph{blocking} scenario, was briefly considered in \cite{stolyar2015pull}. While this policy minimizes the loss probability, it does not achieve an even distribution of the load and involves storing a larger number of tokens: one for each available task slot at a server pool. From a technical perspective, we use a similar methodology to derive a fluid limit in the case of a static admission threshold, but a different methodology is used when the threshold is adjusted over time.

As alluded to above, the most appealing feature of our policy is its capability of adapting the threshold value to unknown and time-varying loads. The problem of adaptation to uncertain demand patterns was addressed in the single-server setting in \cite{goldsztajn2018feedback,goldsztajn2018controlling,mukherjee2017optimal}, which remove the fixed-capacity assumption of the single-server load balancing literature and assume instead that the number of servers can be adjusted on the fly to match the load. However, in these papers the load balancing policy remains the same at all times since the \emph{right-sizing} mechanism for adjusting the number of active servers is sufficient to deal with changes in demand. Mechanisms of this kind had already been studied in the single-server setup to trade off latency and power consumption in microprocessors; see \cite{wierman2012,yao1995scheduling}.

Unlike any of the above papers, this paper considers a token-based dispatching policy for optimizing the overall quality of service in a system of parallel server pools. This policy has a self-learning capability which seamlessly adapts to unknown load values in an online manner, and additionally tracks load variations which are prevalent in practice but rarely accounted for in the load balancing literature.

\subsection{Outline of the paper}
\label{ch2-sub: organization of the paper}

The remainder of the paper is organized as follows. In Section \ref{ch2-sec: model description} we describe our model and the dispatching policy. In Section \ref{ch2-sec: optimal threshold} we establish that this policy is fluid and diffusion optimal if the threshold is suitably chosen. In Section \ref{ch2-sec: learning the optimal threshold} we analyze a control rule for adjusting the threshold to an unknown load and we explain how to tune this control rule for near-optimal performance. Simulations are reported in Section \ref{ch2-sec: simulations}. Several proofs are deferred to Appendices \ref{ch2-app: proofs of various results} and \ref{ch2-app: proof of the fluid limits}.

\section{Model and threshold policy}
\label{ch2-sec: model description}

Consider $n$ parallel and identical server pools with infinitely many servers each. Tasks arrive as a Poisson process with rate $n \lambda$ and are immediately routed to one of the server pools, where service starts at once and lasts an exponentially distributed time of unit mean. The execution times are independent of the number of tasks contending for service at the same server pool, but the quality of service experienced by tasks degrades as the degree of contention increases. Thus, maintaining an even distribution of the load is key for optimizing the overall quality of service.

More specifically, let $X_i$ denote the number of concurrent tasks at server pool $i$ and suppose that we resort to a utility metric $u(X_i) = g(1 / X_i)$ as a proxy for measuring the quality of service experienced by a task assigned to server pool $i$, as function of its resource share. Provided that $g$ is a concave and increasing function, the overall utility $\sum_{i = 1}^n X_iu(X_i)$ is a Schur-concave function of $X = (X_1, \dots, X_n)$, and is thus maximized by balancing the number of tasks among the various server pools.

The vector-valued process $X$ describing the number of tasks at each of the~$n$ server pools constitutes a continuous-time Markov chain when the dispatching decisions are based on the current number of tasks at each server pool. It is, however, more convenient to adopt an aggregate state description, denoting by $\bQ_n(i)$ the number of server pools with at least $i$ tasks, as illustrated in Figure \ref{fig: state variables}. In view of the symmetry of the model, the infinite-dimensional process $\bQ_n = \set{\bQ_n(i)}{i \geq 0}$ also constitutes a continuous-time Markov chain. We will often consider the normalized processes $\bq_n(i) = \bQ_n(i) / n$ and $\bq_n = \set{\bq_n(i)}{i \geq 0}$; the former corresponds to the fraction of server pools with at least $i$ tasks.

\begin{figure}
	\centering
	\includegraphics{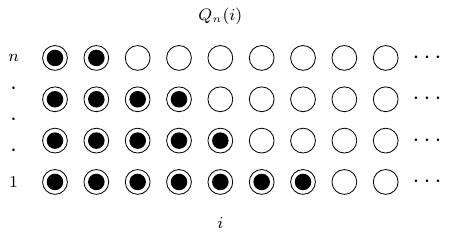}
	\caption{Schematic representation of the system state at a given time. White circles represent servers and black circles represent tasks. Each row corresponds to a server pool, and these are arranged so that the number of tasks increases from top to bottom. The number of tasks in column $i$ equals $Q_n(i)$.}
	\label{fig: state variables}
\end{figure}

\subsection{Threshold-based load balancing policy}
\label{ch2-sub: threshold policy}

All server pools together from an infinite-server system, and thus the total number of tasks in the system in steady state is Poisson distributed with mean $n\lambda$, irrespective of the load balancing policy. In order to motivate our dispatching rule, let us briefly assume that the total number of tasks in the system is actually equal to $n\lambda$ at a given time. This number of tasks would be balanced across the server pools if each server pool had either $\flambda$ or $\clambda$ tasks. This corresponds to the occupancy state $q^*$ defined as
\begin{equation}
	\label{ch2-eq: even distribution of the load}
	q^*(i) \defeq \begin{cases}
		q^*(i) = 1 & \text{if} \quad i \leq \flambda, \\
		q^*(i) = \lambda - \flambda & \text{if} \quad i = \flambda + 1, \\
		q^*(i) = 0 & \text{if} \quad i > \flambda + 1.
	\end{cases}
\end{equation}
If the dispatching rule is JSQ, then it is established in \cite{mukherjee2020asymptotic} that $\bq_n$ has a stationary distribution for each~$n$ and that these stationary distributions converge to the Dirac probability measure concentrated at~$q^*$ as $n \to \infty$. However, this comes at the expense of a significant implementation burden, as observed in Section \ref{ch2-sec: introduction}.

The total number of tasks fluctuates over time, but if tasks are dispatched in a suitable way, then it is possible that almost all the server pools have either $\flambda$ or $\clambda$ tasks most of the time and only the fraction of server pools with each of these numbers of tasks fluctuates. In order to achieve this, we propose a load balancing policy based on an admission threshold $\ell \in \N$ and the current system occupancy; for brevity we also define $h \defeq \ell + 1$. This policy operates as follows.
\begin{itemize}
	\item If $\bq_n(\ell) < 1$, then at least one server pool has strictly less than $\ell$ tasks. In this case every arriving task is assigned to a server pool with strictly less than $\ell$ tasks, selected uniformly at random.
	
	\item If $\bq_n(\ell) = 1$ and $\bq_n(h) < 1$, then all server pools have at least $\ell$ tasks and at least one server pool has exactly $\ell$ tasks. In this case each new task is sent to a server pool chosen uniformly at random among those with exactly $\ell$ tasks.
	
	\item If $\bq_n(h) = 1$, then all server pools have strictly more than $\ell$ tasks. In this case every new task is assigned to a server pool selected uniformly at random.
\end{itemize}

A natural question is if there exists a value of $\ell$ that yields an even distribution of the load. Before addressing this question, we propose a token-based implementation of the threshold-based policy. In this implementation the dispatcher stores at most two tokens per server pool, labeled \emph{green} and \emph{yellow}. For a given server pool, the dispatcher has a green token if the server pool has strictly less than $\ell$ tasks, and a yellow token if the server pool has strictly less than $h$ tasks. Note that both a green and a yellow token are stored if the server pool has strictly less than $\ell$ tasks. When a task arrives, the dispatcher uses the tokens as follows.
\begin{itemize}
	\item In the presence of green tokens, the dispatcher picks one uniformly at random and sends the task to the corresponding server pool; the token is then discarded.
	
	\item If the dispatcher only has yellow tokens, then one is chosen uniformly at random and the task is sent to the corresponding server pool; the token is then discarded.
	
	\item In the absence of tokens a server pool is selected uniformly at random.
\end{itemize}

In order to maintain accurate state information at the dispatcher, the server pools send messages with updates about their status. A server pool with exactly $h$ tasks which finishes one of its tasks will send a yellow message to the dispatcher, in order to generate a yellow token. Similarly, a server pool with exactly $\ell$ tasks, which finishes one of these tasks, will send a green message to the dispatcher to generate a green token. Green messages are also triggered by arrivals when the number of tasks in the server pool receiving the new task is still strictly less than $\ell$ after the arrival. In this way the green token discarded by the dispatcher is replaced.

With this implementation, a given task may trigger at most two messages: one upon arrival and one after leaving the system; i.e., the communication overhead is at most two messages per task. Also, the maximum amount of memory needed at the dispatcher corresponds to $2n$ tokens. In the next section we will establish that our policy is optimal on the fluid and diffusion scales for a suitable threshold. This powerful combination of optimality and low communication overhead resembles the properties of JIQ, as considered in the context of the supermarket model. 

\section{Optimality analysis}
\label{ch2-sec: optimal threshold}

In Section \ref{ch2-sub: static fluid limit} we obtain a fluid model for the threshold-based load balancing policy, based on a system of differential equations. In Section \ref{ch2-sub: fluid-optimal threshold} we use the fluid model to prove that there exist thresholds such that the solutions of the differential equations converge over time to the balanced occupancy state $q^*$. In particular, we prove that $\ell = \flambda$ has this property for all $\lambda > 0$, and is the unique threshold with this property unless $\lambda \in \N$; in the latter case $\ell = \lambda - 1$ is also fluid-optimal. In Section \ref{ch2-sub: diffusion optimality} we prove that $\ell = \flambda$ is diffusion-optimal for all $\lambda > 0$.

\subsection{Fluid limit}
\label{ch2-sub: static fluid limit}

The occupancy processes take values in
\begin{equation*}
	Q \defeq \set{q \in [0, 1]^\N}{q(i + 1) \leq q(i) \leq q(0) = 1\ \text{if}\ i \geq 1\ \text{and}\ \sum_{i = 1}^\infty q(i) < \infty} \subset \R^\N.
\end{equation*}
Recall that the above summation corresponds to the total number of tasks divided by $n$, as illustrated in Figure \ref{fig: state variables}. We endow $\R^\N$ with the product topology and we let $D_{\R^\N}[0, \infty)$ denote the space of c\`adl\`ag functions from $[0, \infty)$ into $\R^\N$, with the topology of uniform convergence over compact sets. All the occupancy processes $\bq_n$ can be constructed on a common probability space as random variables with values in $D_{\R^\N}[0, \infty)$; we outline this construction in Section \ref{ch2-sap: construction on a common probability space}. The following fluid limit holds for any threshold $\ell \in \N$, any random limiting initial condition $q_0$ and any time $T \geq 0$; the proof is provided in Section \ref{ch2-sap: static threshold}. Informally, we prove that the occupancy processes approach solutions of a system of differential equations as $n \to \infty$.

\begin{theorem}
	\label{ch2-the: static fluid limit}
	Suppose that $\bq_n(0) \to q_0$ with probability one in the product topology as $n \to \infty$. Then $\set{\bq_n}{n \geq 1}$ is almost surely relatively compact in $D_{\R^\N}[0, \infty)$. Thus, every subsequence has a further subsequence that converges. Furthermore, the limit of every convergent subsequence is a function $\map{\bq}{[0, \infty)}{Q}$ such that
	\begin{equation}
		\label{ch2-eq: static fluid limit}
		\bq(t, i) = \bq(0, i) + \int_0^t \lambda p_i(\bq(s), \ell)ds - \int_0^ti\left[\bq(s, i) - \bq(s, i + 1)\right]ds
	\end{equation}
	for all $t \geq 0$ and $i \geq 1$. The functions $p_i$ in the above equations are defined as follows.
	\begin{enumerate}
		\item[(a)] If $\bq(\ell) < 1$, then
		\begin{equation*}
			p_i(\bq, \ell) \defeq \begin{cases}
				\frac{\bq(i - 1) - \bq(i)}{1 - \bq(\ell)} & \text{if} \quad 1 \leq i \leq \ell, \\
				0 & \text{if} \quad i \geq h.
			\end{cases}
		\end{equation*}
		
		\item[(b)] If $\bq(\ell) = 1$ and $\bq(h) < 1$, then
		\begin{equation*}
			p_i(\bq, \ell) \defeq \begin{cases}
				\frac{\ell}{\lambda}\left[1 - \bq(h)\right] & \text{if} \quad i = \ell, \\
				1 - \frac{\ell}{\lambda}\left[1 - \bq(h)\right] & \text{if} \quad i = h, \\
				0 & \text{if} \quad i \neq \ell, h.
			\end{cases}
		\end{equation*}
		
		\item[(c)] If $\bq(h) = 1$, then
		\begin{equation*}
			p_i(\bq, \ell) \defeq \begin{cases}
				\frac{h}{\lambda}\left[1 - \bq(h + 1)\right] & \text{if} \quad i = h, \\
				[1 - \frac{h}{\lambda}\left(1 - \bq(h + 1)\right)]\left[\bq(i - 1) - \bq(i)\right] & \text{if} \quad i \geq h + 1, \\
				0 & \text{if} \quad 1 \leq i \leq \ell.
			\end{cases}
		\end{equation*}
	\end{enumerate}
\end{theorem}

The fluid equations \eqref{ch2-eq: static fluid limit} have a simple interpretation. Namely, the derivative of $\bq(i)$ is the rate at which new tasks arrive to server pools with exactly $i - 1$ tasks minus the rate at which tasks leave from server pools with precisely $i$ tasks. The term $i[\bq(i) - \bq(i + 1)]$ corresponds to the cumulative departure rate from server pools with exactly $i$ tasks. This quantity is equal to the total number of tasks in server pools with precisely $i$ tasks and each of these tasks has unit departure rate. The term $p_i(\bq, \ell)$ may be interpreted as the probability that a new task is assigned to a server pool with exactly $i - 1$ tasks in fluid state $\bq$ with threshold $\ell$ in force. Thus, $\lambda p_i(\bq, \ell)$ corresponds to the arrival rate of tasks to server pools with exactly $i - 1$ tasks. The expressions in (a), (b) and (c) correspond to the following situations.
\begin{enumerate}
	\item[(a)] If $\bq(\ell) < 1$, then new tasks are sent to server pools with strictly less than $\ell$ tasks, chosen uniformly at random. Hence, $p_i(\bq, \ell) = 0$ if $i \geq h$ and $p_i(\bq, \ell)$ is the fraction of server pools with exactly $i - 1$ tasks divided by the fraction of server pools with at most $\ell - 1$ tasks if $1 \leq i \leq \ell$.
	
	\item[(b)] If $\bq(\ell) = 1$, then the arrival rate to server pools with precisely $\ell - 1$ tasks must be equal to the departure rate from server pools with exactly $\ell$ tasks, which gives $p_{\ell}(\bq, \ell)$. If $\bq(h) < 1$ and all server pools have at least $\ell$ tasks, then incoming tasks are sent to server pools with exactly $\ell$ tasks. Therefore, $p_h(\bq, \ell) = 1 - p_\ell(\bq, \ell)$ and $p_i(\bq, \ell) = 0$ for all $i \neq \ell, h$.
	
	\item[(c)] If $\bq(h) = 1$, then $p_h(\bq, \ell)$ is determined since the right-hand side of \eqref{ch2-eq: static fluid limit} must be zero for $i = h$. The incoming tasks that are not sent to server pools with exactly $\ell$ tasks, are sent to server pools with $h$ or more tasks, and this happens with probability $1 - p_h(\bq, \ell)$. Since $\bq(i - 1) - \bq(i)$ is the fraction of server pools with $i - 1 \geq h$ tasks, such a server pool is selected with a probability equal to this fraction, following a uniformly random assignment.
\end{enumerate}

It is possible that $p_i(q, \ell) \notin [0, 1]$; e.g., if $\ell > \lambda$, $q(\ell) = 1$ and $q(h) = 0$. However, if $\map{\bq}{[0, \infty)}{Q}$ is a subsequential limit, then $p_i(\bq(t), \ell) \in [0, 1]$ and represents a probability for all $t$ outside a subset of $[0, \infty)$ of zero Lebesgue measure.

The proof of Theorem \ref{ch2-the: static fluid limit} uses a methodology developed in \cite{bramson1998state} to establish the almost sure relative compactness of the sequence of occupancy processes, and that the limit of each convergent subsequence is a function with Lipschitz components. The subsequential limits are then characterized by a careful analysis in neighborhoods of the points where the derivatives of all coordinate functions exist, using the stochastic dynamics of the system. The differential equation \eqref{ch2-eq: static fluid limit} results from this analysis.

\subsection{Fluid-optimal thresholds}
\label{ch2-sub: fluid-optimal threshold}

We say that $\ell$ is \emph{fluid-optimal} if all solutions $\map{\bq}{[0, \infty)}{Q}$ of \eqref{ch2-eq: static fluid limit} satisfy
\begin{equation*}
	\lim_{t \to \infty} \bq(t, i) = q^*(i) \quad \text{for all} \quad i \geq 0.
\end{equation*}
Recall that $q^*$ was defined in \eqref{ch2-eq: even distribution of the load} and corresponds to a balanced load distribution.

\begin{remark}
	Theorem \ref{ch2-the: static fluid limit} implies that solutions to \eqref{ch2-eq: static fluid limit} exist for any given threshold and initial condition. Nonetheless, we do not claim that these solutions are unique.
\end{remark}

In order to identify the fluid-optimal thresholds, we fix $\ell \in \N$ and a solution $\map{\bq}{[0, \infty)}{Q}$ of \eqref{ch2-eq: static fluid limit}. We also introduce the following functions.

\begin{definition}
	\label{ch2-def: total mass and tail functions}
	The functions $\map{\bu, \bv_j}{[0, \infty)}{\R}$ defined as
	\begin{equation*}
		\bu(t) \defeq \sum_{i = 1}^\infty \bq(t, i) \quad \text{and} \quad \bv_j \defeq \sum_{i = j}^\infty \bq(t, i) \quad \text{for all} \quad j \geq 1 
	\end{equation*}
	are called total and tail mass functions, respectively.
\end{definition}

The total mass function $\bu$ quantifies the total number of tasks in the system normalized by the number of server pools. The tail mass function $\bv_j$ has a similar interpretation if we visualize tasks as in the diagram of Figure \ref{fig: state variables}. Specifically, this function measures the total number of tasks, normalized by the number of server pools, that are located in column $j$ of the diagram or further to the right.

Next we state two lemmas that we prove in Appendix \ref{ch2-app: proofs of various results}. These technical lemmas will be used to bound the functions $\bv_h$ and $\bv_{h + 1}$.

\begin{lemma}
	\label{ch2-lem: interchange of infinite sum and derivative}
	Let $\map{\bq}{[0, \infty)}{Q}$ be a solution of \eqref{ch2-eq: static fluid limit}. Then
	\begin{equation*}
		\dot{\bv}_j = \lambda\sum_{i = j}^\infty p_i\left(\bq, \ell\right) - (j - 1)\bq(j) - \bv_j \quad \text{for all} \quad i, j \geq 1,
	\end{equation*}
	almost everywhere on $[0, \infty)$. In particular, $\dot{\bu} = \lambda - \bu$.
\end{lemma}

\begin{lemma}
	\label{ch2-lem: exponential bound}
	Let $\map{f, \varphi}{[a, b)}{\R}$ be locally integrable functions with $a < b \leq \infty$. Suppose that $f$ is absolutely continuous on each finite interval and
	\begin{align*}
		f(x) \leq f(a) + \int_a^x \left[\varphi(y) - f(y)\right]dy \quad \text{for all} \quad x \in [a, b).
	\end{align*}
	Then the following inequality holds:
	\begin{align*}
		f(x) \leq f(a)\e^{-(x - a)} + \int_a^x \varphi(y)\e^{-(x - y)}dy \quad \text{for all} \quad x \in [a, b).
	\end{align*}
\end{lemma}

The equation $\dot{\bu} = \lambda - \bu$, obtained in Lemma \ref{ch2-lem: interchange of infinite sum and derivative}, is just the fluid limit of the total number of tasks in an infinite-server system, which makes obvious sense. We conclude that $\bu(t) \to \lambda$ as $t \to \infty$ since the solution of this equation is
\begin{equation*}
	\bu(t) = \lambda + \left[\bu(0) - \lambda\right]\e^{-t} \quad \text{for all} \quad t \geq 0.
\end{equation*}

The following proposition bounds $\bv_h$ and $\bv_{h + 1}$.

\begin{proposition}
	\label{ch2-prop: asymptotic upper bounds for tail functions}
	Let $x^+ \defeq \max\{x, 0\}$. If $h > \lambda$, then there exists $t_0 \geq 0$ such that
	\begin{subequations}
		\begin{align}
			&\bv_h(t) \leq (\lambda - \ell)^+ + \left[\bv_h(t_0) - (\lambda - \ell)^+\right] \e^{-(t - t_0)}, \label{ch2-seq: upper bound for v_h} \\
			&\bv_{h + 1}(t) \leq \bv_{h + 1}(t_0) \e^{-(t - t_0)}, \label{ch2-seq: upper bound for v_h+1}
		\end{align}
	\end{subequations}
	for all $t \geq t_0$. The last inequality also holds if $h \geq \lambda$, and with $t_0 = 0$.
\end{proposition}

\begin{proof}
	By Lemma \ref{ch2-lem: interchange of infinite sum and derivative}, we have
	\begin{equation*}
		\begin{split}
			\dot{\bv}_{h + 1}(t) &= \left[\lambda - h\left(1 - \bq(t, h + 1)\right)\right]\ind{\bq(t, h) = 1} - h\bq(t, h + 1) - \bv_{h + 1}(t) \\
			&\leq \left[\lambda - h\left(1 - \bq(t, h + 1)\right)\right]^+ - h\bq(t, h + 1) - \bv_{h + 1}(t) \leq (\lambda - h)^+ - \bv_{h + 1}(t).
		\end{split}
	\end{equation*}
	If $h \geq \lambda$, then $(\lambda - h)^+ = 0$ and we get \eqref{ch2-seq: upper bound for v_h+1} for any $t_0 \geq 0$ by Lemma \ref{ch2-lem: exponential bound}.
	
	In order to prove \eqref{ch2-seq: upper bound for v_h}, let us assume that $h > \lambda$, and note that
	\begin{equation*}
		h\bq(t, h) \leq \sum_{i = 1}^h \bq(t, i) \leq \bu(t).
	\end{equation*}
	Since $\bu(t) \to \lambda$ as $t \to \infty$, there exists $t_0 \geq 0$ such that $\bq(t, h) < 1$ for all $t \geq t_0$. We conclude from \eqref{ch2-eq: static fluid limit} and Lemma \ref{ch2-lem: interchange of infinite sum and derivative} that
	\begin{equation*}
		\dot{\bv}_h(t) \leq \left[\lambda - \ell\left(1 - \bq(t, h)\right)\right]^+ - \ell \bq(t, h) - \bv_h(t) \leq (\lambda - \ell)^+ - \bv_h(t) \quad \text{for all} \quad t \geq t_0.
	\end{equation*}
	It follows from this inequality and Lemma \ref{ch2-lem: exponential bound} that \eqref{ch2-seq: upper bound for v_h} holds.
\end{proof}

A consequence of the proposition is that the fraction of server pools with more than $h$ tasks vanishes as $t \to \infty$ if $h \geq \lambda$. We also have the following theorem.

\begin{theorem}
	\label{ch2-the: optimal thresholds}
	The threshold $\ell = \flambda$ is fluid-optimal for all $\lambda > 0$. Moreover, if $\lambda \in \N$ and $\lambda > 0$, then $\ell = \lambda - 1$ is fluid-optimal as well.
\end{theorem}

\begin{proof}
	Suppose that $\ell = \flambda$. By Proposition \ref{ch2-prop: asymptotic upper bounds for tail functions}, we have
	\begin{equation*}
		\limsup_{t \to \infty} \bv_h(t) \leq \lambda - \flambda \quad \text{and} \quad \lim_{t \to \infty} \bv_{h + 1}(t) = 0.
	\end{equation*}	
	The limit superior of $\bq(t, h) = \bv_{h}(t) - \bv_{h + 1}(t)$ as $t \to \infty$ is at most $\lambda - \flambda$. Further,
	\begin{equation*}
		\lim_{t \to \infty} \left[\bq(t, h) + \sum_{i = 1}^\ell \bq(t, i)\right] = \lim_{t \to \infty} \left[\bu(t) - \bv_{h + 1}(t)\right] = \lim_{t \to \infty} \bu(t) = \lambda.
	\end{equation*}
	Since the summation on the left-hand side is at most $\ell = \flambda$ and the limit superior of $\bq(h)$ is upper bounded by $\lambda - \flambda$, we conclude that
	\begin{equation*}
		\lim_{t \to \infty} \sum_{i = 1}^\ell \bq(t, i) = \flambda \quad \text{and} \quad \lim_{t \to \infty} \bq(t, h) = \lambda - \flambda.
	\end{equation*}
	This implies that $\bq(t, i) \to q^*(i)$ as $t \to \infty$ for all $i \geq 0$. A similar argument can be applied in the case where $\lambda \in \Zp$ and the threshold is $\ell = \lambda - 1$.
\end{proof}

The following examples establish that the thresholds indicated in Theorem \ref{ch2-the: optimal thresholds} are the only fluid-optimal thresholds. For larger or smaller thresholds, the examples exhibit stationary solutions $\bq$ of \eqref{ch2-eq: static fluid limit} that are different from $q^*$. In particular, such thresholds are not fluid-optimal since $\bq(t) = \bq(0) \neq q^*$ for all $t \geq 0$.

\begin{example}[Large threshold]
	\label{ch2-ex: large threshold}
	Suppose $\ell > \flambda$ and let $\theta \in (0, 1)$ be a solution of
	\begin{equation}
		\label{ch2-eq: theta of example 1}
		\frac{\lambda}{x} = \sum_{i = 1}^{\ell} \frac{\ell!}{(\ell - i)!}\left(\frac{1 - x}{\lambda}\right)^{i - 1}.
	\end{equation}
	Such a solution always exists since the right-hand side is continuous as a function of $x \in [0, 1]$ and equals $\ell > \lambda$ at $x = 1$. Define $q \in Q$ such that $q(i) = 0$ for all $i > \ell$ and
	\begin{equation*}
		q(\ell - i) - q\left(\ell - (i - 1)\right) = \frac{\ell!}{(\ell - i)!}\left(\frac{1 - \theta}{\lambda}\right)^i\theta \quad \text{for all} \quad 0 \leq i \leq \ell.
	\end{equation*}
	In order to see that $q$ indeed lies in $Q$, observe that $q(\ell) = \theta$ and
	\begin{equation*}
		1 = \theta + \sum_{i = 1}^\ell \frac{\ell!}{(\ell - i)!}\left(\frac{1 - \theta}{\lambda}\right)^i \theta = \sum_{i = 0}^\infty \left[q(i) - q(i + 1)\right] = q(0),
	\end{equation*}
	where the first equality is due to \eqref{ch2-eq: theta of example 1}. It is possible to check that $q$ is an equilibrium point of \eqref{ch2-eq: static fluid limit}. Therefore, the constant function $\map{\bq}{[0, \infty)}{Q}$ such that $\bq(t) = q$ for all $t \geq 0$ is a stationary solution. As noted earlier, this implies that every $\ell > \flambda$ is not a fluid-optimal threshold because $q \neq q^*$.
\end{example}

The above equilibrium point $q$ corresponds to a suboptimal distribution of the load where the fraction of server pools with exactly $i$ tasks is positive for all $0 \leq i \leq \ell$, as depicted in Figure \ref{ch2-fig: suboptimal equilibria}. Note that $\pi(j) \defeq q(j) - q(j + 1)$ is the stationary distribution of an Erlang-B system with $\ell$ servers and offered traffic $\lambda / (1 - \theta)$. This is reasonable since each server pool behaves as a blocking system with $\ell$ servers except when all the server pools have at least $\ell$ tasks. The fact that the offered traffic is larger than $\lambda$ also makes sense, because server pools with less than $\ell$ tasks have a larger arrival rate when some server pools have at least $\ell$ tasks and are thus blocked. Moreover, if a fraction $1 - \theta = 1 - q(\ell)$ of the server pools receive tasks at rate $\lambda / (1 - \theta)$ and the other server pools are blocked, then the overall load is $\lambda$.

If the load is distributed according to the occupancy state $q$ of the example, then the maximum number of tasks across the server pools is $\ell > \flambda$. If the goal is to avoid concentrations of tasks, then this is near-optimal when $\ell$ is close to $\flambda$. The most problematic situations arise instead when the threshold is lower than the optimal value. In these cases all server pools have more than $\ell$ tasks most of the time, because the threshold is smaller than the load. As a result, the system resorts to uniformly random routing very often, which is known to be highly inefficient. In the large-scale limit this translates into the number of tasks across server pools being unbounded.

\begin{example}[Small threshold]
	\label{ch2-ex: small threshold}
	Suppose that $\ell < \flambda$ and $\lambda \notin \Zp$, or alternatively that $\ell < \lambda - 1$ and $\lambda \in \Zp$; note that $h < \lambda$ in either case. Also, let $\theta \in (0, 1)$ solve
	\begin{equation}
		\label{ch2-eq: theta of example 2}
		\sum_{i = h}^\infty \frac{h!}{i!}\left[\lambda - h(1 - x)\right]^{i - h}(1 - x) = 1.
	\end{equation}
	A solution always exists since the left-hand side is strictly larger than one for $x = 0$ and equal to zero for $x = 1$. We define $q \in Q$ such that $q(i) = 1$ for all $1 \leq i \leq h$ and
	\begin{equation*}
		q(i) - q(i + 1) = \frac{h!}{i!}\left[\lambda - h(1 - \theta)\right]^{i - h} (1 - \theta) \quad \text{for all} \quad i \geq h.
	\end{equation*}
	By setting $i = h$ in the last equation, we observe that $1 - \theta$ is the fraction of server pools with at most $h$ tasks. It is possible to check that $q$ is an equilibrium solution of \eqref{ch2-eq: static fluid limit}. Since $q \neq q^*$, we conclude as in Example \ref{ch2-ex: large threshold} that every threshold $\ell < \flambda$ if $\lambda \notin \N$, or $\ell < \lambda - 1$ if $\lambda \in \N$, is not fluid-optimal.
\end{example}

If the load is distributed according to the above occupancy state $q$, then all server pools have at least $h$ tasks and for each $i$ there exists a positive fraction of pools with at least $i$ tasks; this is illustrated in Figure \ref{ch2-fig: suboptimal equilibria}. Note that $\pi(j) \defeq q(j) - q(j + 1)$ is the stationary distribution of the birth-death process with death rate $j$ at state $j$ and birth rate $\lambda - h(1 - \theta)$ at each state. Informally, this can be interpreted as follows. The fraction of server pools with at least $h$ tasks equals one most of the time. This requires that tasks leaving a server pool with $\ell$ tasks are quickly replaced by new tasks, which happens at fluid rate $\lambda p_h(q, \ell) = h(1 - \theta) < \lambda$. The remaining tasks, which arrive at fluid rate $\lambda - h(1 - \theta)$, are sent to server pools that are selected uniformly at random. Therefore, server pools with at least $h$ tasks behave as independent birth-death processes with the birth rate and death rate indicated above.

\begin{figure}
	\centering
	\begin{subfigure}{0.49\columnwidth}
		\centering
		\includegraphics[width = \columnwidth]{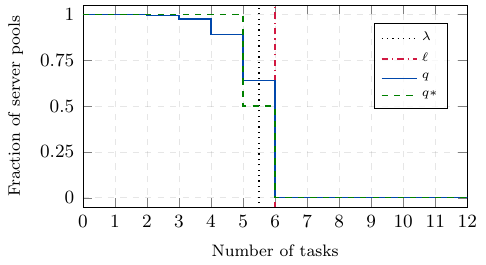}
	\end{subfigure}%
	\hfill
	\begin{subfigure}{0.49\columnwidth}
		\centering
		\includegraphics[width = \columnwidth]{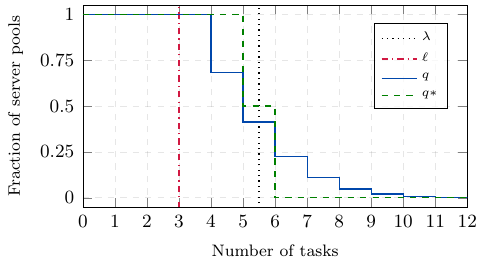}
	\end{subfigure}
	\caption{Illustrations of the equilibrium solutions $q$ computed in Examples \ref{ch2-ex: large threshold} and \ref{ch2-ex: small threshold} for $\lambda = 5.5$ and the thresholds indicated in the plots. The left plot corresponds to the case of a large threshold, whereas the right plot corresponds to the case of a small threshold.}
	\label{ch2-fig: suboptimal equilibria}
\end{figure}

The following corollary is a consequence of Theorem \ref{ch2-the: optimal thresholds} and the two examples.

\begin{corollary}
	If $\lambda \notin \Zp$, then $\flambda$ is the unique fluid-optimal threshold. If $\lambda \in \Zp$, then there are two fluid-optimal thresholds: $\lambda - 1$ and $\lambda$.
\end{corollary}

\subsection{Diffusion-scale optimality}
\label{ch2-sub: diffusion optimality}

We have proved that $\ell = \flambda$ is a fluid-optimal threshold for all $\lambda > 0$. Next we further prove that the threshold-based load balancing policy has the same diffusion limit as JSQ for the latter threshold value. Recall that \cite{menich1991optimality,sparaggis1993extremal} showed that JSQ is the optimal nonanticipating load balancing policy, and that this result is nonasymptotic. Hence, the diffusion limit of JSQ is the best possible diffusion limit. This limit was derived in \cite{mukherjee2020asymptotic} and has different forms for $\lambda \notin \N$ and $\lambda \in \N$. For completeness we also state both diffusion limits here. Then we prove that the same limits hold for the threshold policy using a stochastic coupling argument.

\subsubsection{Diffusion limit for $\lambda \notin \N$}

Consider the following processes:
\begin{subequations}
	\label{ch2-eq: scaled processes in the non-integer case}
	\begin{align}
		&\bar{\bY}_n \defeq \sum_{i = 1}^\flambda \frac{n - \bQ_n(i)}{\sqrt{n}}, \\
		&\bar{\bZ}_n \defeq \frac{\bQ_n(\clambda) - n(\lambda -\flambda)}{\sqrt{n}}, \\
		&\bar{\bQ}_n(i) \defeq \frac{\bQ_n(i)}{\sqrt{n}} \quad \text{for all} \quad i > \clambda.
	\end{align}
\end{subequations}
The following theorem holds for JSQ and the threshold-based policy with $\ell = \flambda$.

\begin{theorem}
	\label{ch2-the: diffusion limit in the non-integer case}
	Suppose that $\lambda \notin \Zp$ and there exists $M \geq \clambda$ such that
	\begin{equation*}
		\left(\bar{\bY}_n(0), \bar{\bZ}_n(0), \bar{\bQ}_n(0, \clambda + 1), \dots, \bar{\bQ}_n(0, M)\right) \Rightarrow \left(0, \bar{Z}, 0, \dots, 0\right)
	\end{equation*}
	in $\R^{M - \flambda + 2}$ as $n \to \infty$ and $\bQ_n(0, M + 1) = 0$ for all large enough $n$. Then the following two statements hold.
	\begin{enumerate}
		\item[(a)] The process $\bar{\bY}_n$ converges weakly to the identically zero process and the same holds for the processes $\bar{\bQ}_n(i)$ for all $\clambda + 1 \leq i \leq M$.
		
		\item[(b)] The stochastic process $\bar{\bZ}_n$ converges weakly in $D_\R[0, \infty)$ as $n \to \infty$ to the Ornstein-Uhlenbeck process that solves $d\bZ(t) = -\bZ(t) dt + \sqrt{2\lambda}d\bW(t)$, where $\bW$ denotes a standard Wiener process.
	\end{enumerate}
\end{theorem}

This theorem corresponds to \cite[Theorem 2]{mukherjee2020asymptotic}. Informally, the theorem describes the stochastic fluctuations of a large system around the optimal occupancy state $q^*$. In particular, the number of server pools with $\clambda$ tasks is $n(\lambda - \flambda) + O(\sqrt{n})$ and only $O(\sqrt{n})$ server pools have fewer than $\flambda$ or more than $\clambda$ tasks.

\subsubsection{Diffusion limit for $\lambda \in \N$}

Consider now the processes:
\begin{subequations}
	\label{ch2-eq: scaled processes in the integer case}
	\begin{align}
		&\hat{\bY}_n \defeq \sum_{i = 1}^{\lambda - 1} \frac{n - \bQ_n(i)}{\sqrt{n}}, \\
		&\hat{\bZ}_n \defeq \sum_{i = 1}^{\lambda} \frac{n - \bQ_n(i)}{\sqrt{n}}, \\
		&\hat{\bQ}_n(i) \defeq \frac{\bQ_n(i)}{\sqrt{n}} \quad \text{for all} \quad i > \lambda.
	\end{align}
\end{subequations}
The following theorem holds for JSQ and the threshold-based policy with $\ell = \lambda$.

\begin{theorem}
	\label{ch2-the: diffusion limit in the integer case}
	Suppose that $\lambda \in \Zp$ and there exists $M \geq \lambda$ such that
	\begin{equation*}
		\left(\hat{\bY}_n(0), \hat{\bZ}_n(0), \hat{\bQ}_n(0, \lambda + 1), \dots, \hat{\bQ}_n(0, M)\right) \Rightarrow \left(0, \hat{Z}, \hat{Q}(\lambda + 1), \dots, \hat{Q}(M)\right)
	\end{equation*}
	in $\R^{M - \lambda + 2}$ as $n \to \infty$ and $\bQ_n(0, M + 1) = 0$ for all large enough $n$. The stochastic process $(\hat{\bZ}_n, \hat{\bQ}_n(\lambda + 1), \dots, \hat{\bQ}_n(M))$ converges weakly in $D_{\R^{M - \lambda + 2}}[0, \infty)$ to the unique solution of the following stochastic integral equations:
	\begin{align*}
		&\hat{\bZ}(t) = \hat{Z} + \sqrt{2\lambda}\bW(t) + \bV_1(t) - \int_0^t \left[\hat{\bZ}(s) + \lambda\hat{\bQ}(s, \lambda + 1)\right]ds, \\
		&\hat{\bQ}(t, \lambda + 1) = \hat{Q}(\lambda + 1) + \bV_1(t) - \int_0^t (\lambda + 1) \left[\hat{\bQ}(s, \lambda + 1) - \hat{\bQ}(s, \lambda + 2)\right]ds, \\
		&\hat{\bQ}(t, i) = \hat{Q}(i) - \int_0^t i\left[\hat{\bQ}(s, i) - \hat{\bQ}(s, i + 1)\right]ds \quad \text{for all} \quad i = \lambda + 2, \dots, M - 1, \\
		&\hat{\bQ}(t, M) = \hat{Q}(M) - \int_0^t M\hat{\bQ}(s, M)ds,
	\end{align*}
	where $\bW$ is a standard Wiener process and $\bV_1$ is the unique nondecreasing and nonnegative process in $D_{\R}[0, \infty)$ such that
	\begin{equation*}
		\bV_1(0) = 0 \quad \text{and} \quad \int_0^t \ind{\hat{\bZ}(s) > 0}d\bV_1(s) = 0.
	\end{equation*}
\end{theorem}

This theorem corresponds to \cite[Theorem 3]{mukherjee2020asymptotic} and is similar to the diffusion limit of JSQ derived in \cite{eschenfeldt2018join} for the single-server case. As in the $\lambda \notin \N$ case, the theorem describes the fluctuations of a large system around the optimal occupancy state $q^*$. In particular, there are $O(\sqrt{n})$ server pools with fewer than $\lambda$ or more than $\lambda$ tasks, so most of the server pools have exactly $\lambda$ tasks.

\subsection{Proof of the diffusion limits}
\label{ch2-sap: proof of the diffusion limits}

The proofs of Theorems \ref{ch2-the: diffusion limit in the non-integer case} and \ref{ch2-the: diffusion limit in the integer case} are based on a stochastic coupling over a finite interval of time between a system that uses the threshod policy and one that uses JSQ. We prove that both systems are suitably equivalent if incoming tasks are discarded when all the server pools have at least $\flambda + 1$ tasks. Then we consider nonblocking systems and we show that the probability that a task is sent to a server pool with at least $\flambda + 1$ tasks vanishes as $n \to \infty$, which implies that the nonblocking systems are equivalent in the limit. Theorems \ref{ch2-the: diffusion limit in the non-integer case} and \ref{ch2-the: diffusion limit in the integer case} follow from this fact and the corresponding diffusion limits for JSQ derived in \cite{mukherjee2020asymptotic}.

Consider two blocking systems with $n$ server pools each, such that each server pool can have at most $\flambda + 1$ tasks. The state of the system that uses JSQ is $\bQ_n^1$ and the state of the system with the threshold-based policy is $\bQ_n^2$, where the threshold is $\ell = \flambda$ for all $n$. All tasks arrive simultaneously at both systems according to a single Poisson process of intensity $n \lambda$. We also use a single Poisson process for counting the potential departures from both systems. The intensity of this process is the maximum number of tasks that the systems can have, which is $n(\flambda + 1)$.

In order to couple the departures, we define 
\begin{equation*}
	\ba_n^k \defeq \frac{\bQ_n^k(\flambda + 1)}{n} \quad \text{and} \quad \bb_n^k \defeq \frac{1}{n(\flambda + 1)}\sum_{i = 1}^{\flambda + 1} \bQ_n^k(i) \quad \text{for each} \quad k \in \{1, 2\}.
\end{equation*}
Note that $n(\flambda + 1)\ba_n^k$ is the number of tasks in server pools with exactly $\flambda + 1$ tasks, whereas the total number of tasks in the system is $n(\flambda + 1)\bb_n^k$. At each potential departure time $\tau$, we draw a single number $U \in (0, 1)$ uniformly at random. We establish whether a task leaves system $k$, and from which pool, as follows.
\begin{itemize}
	\item If $U < \ba_n^k(\tau)$, then a we select a server pool with exactly $\flambda + 1$ tasks uniformly at random and we remove one task.
	
	\item If $\ba_n^k(\tau) \leq U < \bb_n^k(\tau)$, then we select a server pool with at most $\flambda$ tasks uniformly at random and we remove one task.
	
	\item If $\bb_n^k(\tau) \leq U$, then no task leaves the system.
\end{itemize}
Since potential departures occur at rate $n(\flambda + 1)$, tasks leave system $k$ at a rate $n(\flambda + 1)\bb_n^k$ that is equal to the total number of tasks in the system. Furthermore, a task leaves a server pool with exactly $\flambda + 1$ tasks at a rate $n(\flambda + 1)\ba_n^k$ that is equal to the number of tasks in such server pools.

\begin{proposition}
	\label{ch2-prop: coupling}
	Fix $n$ and suppose that $\ell = \flambda$. The vector-valued processes
	\begin{equation*}
		\bX_n^1 \defeq \left(\sum_{i = 1}^\flambda \bQ_n^1(i), \bQ_n^1(\flambda + 1)\right) \quad \text{and} \quad \bX_n^2 \defeq \left(\sum_{i = 1}^\flambda \bQ_n^2(i), \bQ_n^2(\flambda + 1)\right)
	\end{equation*}
	have the same law if the initial conditions $\bQ_n^1(0)$ and $\bQ_n^2(0)$ have the same law.
\end{proposition}

\begin{proof}
	Using the coupling described above, we may construct the two systems on a common probability space, with the same initial conditions. We will establish that $\bX_n^1(t) = \bX_n^2(t)$ for all $t \geq 0$ with probability one. For this purpose we consider the random times $0 = \tau_0 < \tau_1 < \dots$ of arrivals or potential departures and we proceed by induction. By construction $\bX_n^1(\tau_0) = \bX_n^2(\tau_0)$. Suppose that this holds at $\tau_i$ for some $i \geq 0$ and let us prove that then it holds at $\tau_{i + 1}$ as well.
	
	First assume that $\tau_{i + 1}$ is an arrival epoch. If $\bX_n^1(\tau_i, 1) = \bX_n^2(\tau_i, 1) < n\flambda$, then both systems have some server pool with less than $\flambda$ tasks when a new tasks arrives at $\tau_{i + 1}$. Both systems assign the new task to a server pool with less than $\flambda$ tasks and only the first entries of $\bX_n^1$ and $\bX_n^2$ change, increasing by one in both cases. Suppose now that $\bX_n^1(\tau_i, 1) = \bX_n^2(\tau_i, 1) = n\flambda$ and $\bX_n^1(\tau_i, 2) = \bX_n^2(\tau_i, 2) < n$. When a new task arrives at $\tau_{i + 1}$, all the server pools in both systems are either completely full or have space for just one more task. The new task is sent to one of the idle servers in both systems and only the second entries of $\bX_n^1$ and $\bX_n^2$ change, increasing by one in both cases. In the remaining case both systems are completely full right before $\tau_{i + 1}$, so the new task has to be discarded and $\bX_n^1$ and $\bX_n^2$ do not change.
	
	Suppose now that $\tau_{i + 1}$ is a potential departure epoch. Since $\bX_n^1(\tau_i) = \bX_n^2(\tau_i)$, then $\ba_n^1(\tau_i) = \ba_n^2(\tau_i)$ and $\bb_n^1(\tau_i) = \bb_n^2(\tau_i)$. As a result, $\bX_n^1(\tau_{i + 1}) = \bX_n^2(\tau_{i + 1})$.
\end{proof}

In order to simplify the exposition, we have assumed that both systems have no server pools with more than $\flambda + 1$ tasks at all times. Nonetheless, the stochastic coupling and the proof of the proposition can be extended to a setting where server pools are allowed to have more than $\flambda + 1$ tasks at time zero and incoming tasks are discarded whenever all the server pools have at least $\flambda + 1$ tasks.

Consider the processes $(\bar{\bY}\vphantom{Y}_n^k, \bar{\bZ}\vphantom{Z}_n^k)$ and $(\hat{\bZ}\vphantom{Z}_n^k, \hat{\bQ}\vphantom{Q}_n^k(\lambda + 1))$ associatd with blocking systems where the number of tasks in a server pool may exceed $\flambda + 1$ at time zero. The law of these processes is independent of $k$. The next proposition implies that if these processes have a limit in distribution, then the same limit holds for the processes associated with nonblocking systems.

\begin{proposition}
	\label{ch2-prop: asymptotic probability of random routing}
	Suppose that the assumptions of Theorem \ref{ch2-the: diffusion limit in the non-integer case} or Theorem \ref{ch2-the: diffusion limit in the integer case} hold and fix $T > 0$. The probability that one or more tasks are dispatched to a server pool with at least $\flambda + 1$ tasks within the interval $[0, T]$ vanishes as $n \to \infty$ when JSQ or the threshold-based policy with $\ell = \flambda$ are used.
\end{proposition}

\begin{proof}
	The same arguments apply to both policies, thus we drop $k$. Let $\bX_n$ denote the total number of tasks in the system and note that Theorems \ref{ch2-the: diffusion limit in the non-integer case} and \ref{ch2-the: diffusion limit in the integer case} assume that there exists a constant $M$ such that
	\begin{equation*}
		\bX_n(0) = \sum_{i = 1}^M \bQ_n(0, i) \quad \text{for all large enough} \quad n.
	\end{equation*}
	The assumptions of the theorems further imply that $\bX_n(0) / n \Rightarrow \lambda$ as $n \to \infty$.
	
	The process $\bX_n$ represents the total number of tasks in an infinite-server system with Poisson arrivals at rate $n\lambda$ and unit-mean exponential service times. Hence, we have $\bX_n / n \Rightarrow \bx$ in $D_{\R}[0, \infty)$ as $n \to \infty$, where
	\begin{equation*}
		\dot{\bx}(t) = \lambda - \bx(t) \quad \text{for all} \quad t \geq 0 \quad \text{almost surely},
	\end{equation*}
	and $\bx(0) = \lambda$ almost surely since $\bX_n(0) / n \Rightarrow \bx(0)$. Then $\bx(t) = \lambda$ for all $t \geq 0$ with probablity one, and the above fluid limit can be expressed as follows:
	\begin{equation*}
		\lim_{n \to \infty} P\left(\sup_{t \in [0, T]} \left|\frac{\bX_n(t)}{n} - \lambda\right| \leq \varepsilon\right) = 1 \quad \text{for all} \quad \varepsilon > 0 \quad \text{and} \quad T \geq 0.
	\end{equation*}
	
	If we fix $0 < \delta < \flambda + 1 - \lambda$, then
	\begin{equation*}
		\lim_{n \to \infty} P\left(\sup_{t \in [0, T]} \bX_n(t) \leq n\left(\flambda + 1 - \delta\right)\right) = 1 \quad \text{for all} \quad T \geq 0.
	\end{equation*}
	For both policies, tasks are not sent to server pools with at least $\flambda + 1$ tasks unless all server pools have at least that number of tasks. This requires that  the total number of tasks is at least $n \left(\flambda + 1\right)$. If $\tau_n$ is the first time that a task is sent to a server pool with $\flambda + 1$ tasks or more in the system with $n$ server pools, then the inequality inside of the above probability sign implies that $\tau_n > T$. We conclude that
	\begin{equation*}
		\lim_{n \to \infty} P\left(\tau_n \leq T\right) = 0 \quad \text{for all} \quad T \geq 0.
	\end{equation*}
	This completes the proof.
\end{proof}

While tasks are not sent to server pools with at least $\flambda + 1$ tasks, the law of $\bQ_n^k(i)$ does not depend on $k$ for $i \geq \flambda + 2$ since the evolution of $\bQ_n^k(i)$ only depends on the departures from server pools with exactly $i$ tasks. Hence, the weak limits of
\begin{align*}
	&\left(\bar{\bY}_n^k, \bar{\bZ}_n^k, \bar{\bQ}_n^k\left(\flambda + 1\right), \dots, \bar{\bQ}_n^k\left(M\right)\right) \quad \text{if} \quad \lambda \notin \N, \\
	&\left(\hat{\bZ}_n^k, \bar{\bQ}_n^k\left(\lambda + 1\right), \dots, \bar{\bQ}_n^k\left(M\right)\right) \quad \text{if} \quad \lambda \in \N,
\end{align*}
do not depend on $k$. As noted earlier, Theorems \ref{ch2-the: diffusion limit in the non-integer case} and \ref{ch2-the: diffusion limit in the integer case} hold for JSQ by \cite{mukherjee2020asymptotic}, so these theorems also hold for the threshold policy with $\ell = \flambda$.

\section{Learning the optimal threshold}
\label{ch2-sec: learning the optimal threshold}

In Section \ref{ch2-sec: optimal threshold} we showed that our threshold-based policy is fluid and diffusion optimal provided that $\ell = \flambda$. However, these optimality properties critically rely on the threshold being strictly equal to $\flambda$, as was shown by Examples \ref{ch2-ex: large threshold} and \ref{ch2-ex: small threshold}. Furthermore, in actual system deployments discrepancies between an a priori chosen threshold and the optimal value $\flambda$ may occur due to the following reasons.
\begin{itemize}
	\item It can be difficult to estimate $\lambda$ in advance and a slightly inaccurate estimate may result in a wrong choice of the threshold. The worst repercussions in terms of performance occur when $\lambda$ is underestimated and a low threshold is chosen, as explained right before Example \ref{ch2-ex: small threshold}.
	
	\item The load $\lambda$ can change over time due to fluctuations in the demand. These fluctuations can result in a mismatch between $\ell$ and $\flambda$, even if $\ell = \flambda$ initially.
\end{itemize}

\begin{remark}
	We have adopted the common assumption of unit-mean task durations, which amounts to a convenient choice of time unit. In view of this, it is worth noting that the optimal threshold is determined by the offered load, rather than the arrival rate of tasks. Namely, if task durations had mean $1 / \mu$, then the optimal threshold would be $\ell = \floor{\rho}$, with $\rho \defeq \lambda / \mu$ and $n \rho = n \lambda / \mu$ the offered load. In particular, it is the offered load that has to be estimated rather than the arrival rate of tasks, which exacerbates the issues mentioned above. The results in this paper easily generalize to any service rate $\mu$ without changing the control rule to be described below, which is designed to track the offered load rather than the arrival rate of tasks.
\end{remark}

Next we introduce a control rule for adjusting the threshold over time so as to learn the optimal threshold value when $\lambda$ is unknown. We analyze this control rule through a fluid model that we describe in Section \ref{ch2-sub: fluid systems} and justify in Section \ref{ch2-sub: dynamic fluid limit} through a fluid limit. In Section \ref{ch2-sub: convergence of the threshold} we prove that the dynamic threshold of the fluid model always reaches an equilibrium, and in Section \ref{ch2-sub: parameter setting} we explain how to tune the control rule so that the equilibrium threshold is always near-optimal. Further, we show that this tuning yields an optimal equilibrium threshold in most cases. The time required for the threshold to settle is analyzed in Section \ref{ch2-sub: optimality and time until settling}.

\subsection{Learning scheme}
\label{ch2-sub: learning rule}

In order to achieve optimality, we need to actively learn the optimal threshold value. For this purpose we introduce a control rule for adjusting the threshold in an online manner. Let us denote the threshold of a system with $n$ server pools at time $t$ by $\ell_n(t)$, which is now \emph{time-dependent}, and as before, let $h_n(t) \defeq \ell_n(t) + 1$ for brevity. The control rule depends on a parameter $\alpha \in (0, 1)$ and adjusts the threshold only at arrival epochs, right after a new task has been dispatched. If an arrival occurs at time $\tau$, then the threshold is adjusted as follows.
\begin{itemize}
	\item The threshold is increased by one if the number of server pools with at least $h_n$ tasks, measured right before time $\tau$, is greater than or equal to $n - 1$.
	
	\item The threshold is decreased by one if the fraction of server pools with at least $\ell_n$ tasks, measured right before time $\tau$, is smaller than or equal to $\alpha$.
	
	\item Otherwise, the threshold remains unchanged.
\end{itemize}

Note that this control rule only relies on knowledge of the tokens that are used for dispatching the incoming tasks. Specifically, the threshold is increased if and only if the number of yellow tokens is zero when a task arrives or would be zero after dispatching the task. Also, the threshold is decreased if and only if the number of green tokens is larger than or equal to $(1 - \alpha) n$ right before an arrival.

\subsection{Fluid systems}
\label{ch2-sub: fluid systems}

Suppose $\lambda$ is unknown, either because it was not possible to estimate the offered load in advance or because it recently changed. As tasks arrive to the system, the control rule adjusts the threshold in steps of one unit, in search of the optimal value. Next we provide a fluid model for the occupancy state and dynamic threshold of the system, which will be used to establish that the threshold updates eventually cease, with the threshold reaching an equilibrium. We use the term \emph{fluid system} to refer to the occupancy state and dynamic threshold in the fluid model.

\begin{definition}
	\label{ch2-def: fluid system}
	Let $\eta \in \N \cup \{\infty\}$ and consider sequences
	\begin{equation*}
		\set{\tau_j \geq 0}{0 \leq j < \eta} \quad \text{and} \quad \set{\ell_j \in \N}{0 \leq j < \eta}
	\end{equation*}
	of strictly increasing times and thresholds, respectively. Suppose that $\tau_0 = 0$ and let
	\begin{equation*}
		\tau_\eta \defeq \infty \quad \text{if} \quad \eta < \infty \quad \text{and} \quad \tau_\eta \defeq \lim_{j \to \infty} \tau_j \quad \text{if} \quad \eta = \infty.
	\end{equation*}
	We define a piecewise constant function $\map{\ell}{[0, \tau_\eta)}{\N}$ by
	\begin{equation*}
		\ell(t) \defeq \ell_j \quad \text{for all} \quad t \in [\tau_j, \tau_{j + 1}) \quad \text{and} \quad 0 \leq j < \eta;
	\end{equation*}
	as usual, we let $h_j \defeq \ell_j + 1$ and $h(t) \defeq \ell(t) + 1$. Given $\map{\bq}{[0, \tau_\eta)}{Q}$, we say that $\bs \defeq (\bq, \ell)$ is a fluid system if the following three conditions hold.
	\begin{enumerate}
		\item[(a)] The coordinate functions $\bq(i)$ are absolutely continuous on finite intervals and
		\begin{equation}
			\label{ch2-eq: fluid dynamics revisited}
			\dot{\bq}(i) = \lambda p_i(\bq, \ell_j) - i\left[\bq(i) - \bq(i + 1)\right] \quad \text{for all} \quad i \geq 1
		\end{equation}
		almost everywhere on $[\tau_j, \tau_{j + 1})$ for all $0 \leq j < \eta$.
		
		\item[(b)] $\bq(t, \ell(t)) \geq \alpha$ for all $t \in [0, \tau_\eta)$ and $\bq(h) < 1$ almost everywhere on $[0, \tau_\eta)$.
		
		\item[(c)] $\bq\left(\tau_{j + 1}, \ell_j\right) = \alpha$ or $\bq\left(\tau_{j + 1}, h_j\right) = 1$ for all $0 \leq j < \eta - 1$. Also,
		\begin{equation*}
			\ell_{j + 1} = \ell_j - 1 \quad \text{if} \quad \bq\left(\tau_{j + 1}, \ell_j\right) = \alpha \quad \text{and} \quad \ell_{j + 1} = \ell_j + 1 \quad \text{if} \quad \bq\left(\tau_{j + 1}, h_j\right) = 1.
		\end{equation*}
	\end{enumerate}
\end{definition}

A fluid system consists of a function $\bq$, which represents the evolution of the occupancy state, and a piecewise constant function $\ell$, which represents the dynamic threshold. Between $\tau_j$ and $\tau_{j + 1}$ the threshold is constant, equal to $\ell_j$, and the system behaves according to the differential equation of Theorem \ref{ch2-the: static fluid limit}. Also, the dynamic threshold is adjusted at the times $\tau_j$ according to the control rule explained above: it increases when $\bq(h)$ reaches one and decreases when $\bq(\ell)$  reaches $\alpha$.

The possibly finite time $\tau_\eta$ accounts for the possibility of infinitely many threshold updates in finite time. We prove however that this in fact cannot happen. For this purpose we resort to the total mass function $\bu$. We have
\begin{equation}
	\label{ch2-eq: total mass}
	\bu(t) = \lambda + [\bu(0) - \lambda]\e^{-t} \quad \text{for all} \quad t \in [0, \tau_\eta)
\end{equation}
since $\dot{\bu} = \lambda - \bu$ as in Lemma \ref{ch2-lem: interchange of infinite sum and derivative}. The following proposition establishes that the threshold of a fluid system cannot change infinitely many times in finite time.

\begin{proposition}
	\label{ch2-prop: infinitely many updates cannot occur in finite time}
	All fluid systems are such that $\tau_\eta = \infty$.
\end{proposition}

\begin{proof}
	Consider a fluid system with $\eta = \infty$, otherwise $\tau_\eta = \infty$ by definition. It follows from \eqref{ch2-eq: total mass} that $\bu$ is upper bounded by some constant $M \geq 0$. Since the infinite sequence $\bq(t) \in Q$ is nonincreasing for each given $t$, we have
	\begin{equation*}
		\alpha \leq \bq(t, \ell_j) \leq \frac{1}{\ell_j} \sum_{i = 1}^{\ell_j} \bq(t, i) \leq \frac{\bu(t)}{\ell_j} \leq \frac{M}{\ell_j} \quad \text{for all} \quad t \in [\tau_j, \tau_{j + 1}) \quad \text{and} \quad j \geq 0.
	\end{equation*}
	We conclude that the set $\set{\ell_j}{j \geq 0}$ is bounded, and thus the set
	\begin{equation*}
		\calL = \set{l \in \N}{\ell_j = l\ \text{for infinitely many}\ j}
	\end{equation*}
	is nonempty and bounded.
	
	Note that there exists $j_0 \geq 0$ such that $\ell_j \leq m \defeq \max \calL$ for all $j \geq j_0$. Indeed, otherwise $\set{\ell_j}{j \geq 0}$ takes values in the finite set $\{m + 1, \dots, \floor{M / \alpha}\}$ infinitely often. This implies that the threshold takes some given value in the latter set infinitely often, which leads to a contradiction since the latter set and $\calL$ are disjoint.
	
	Fix $j > j_0$ such that $\ell_j = m$. By \eqref{ch2-eq: fluid dynamics revisited} we know that
	\begin{equation*}
		\dot{\bq}(m) = \lambda p_m(\bq, \ell_j) - m\left[\bq(m) - \bq(m + 1)\right] \geq -m
	\end{equation*}
	almost everywhere on $[\tau_j, \tau_{j + 1})$. The thresholds $\ell_{j - 1}$ and $\ell_{j + 1}$ are equal to $m - 1$ by definition of $j_0$. It follows that $\bq(\tau_j, m) = 1$ and $\bq\left(\tau_{j + 1}, m\right) = \alpha$, which implies that
	\begin{equation*}
		\alpha = \bq(\tau_{j + 1}, m) \geq \bq(\tau_j, m) - m(\tau_{j + 1} - \tau_j) = 1 - m(\tau_{j + 1} - \tau_j).
	\end{equation*}
	
	By definition, there exist infinitely many indexes $j > j_0$ such that $\ell_j = m$, and we proved that $\tau_{j + 1} - \tau_j \geq (1 - \alpha) / m$ for these indexes. Thus, $\tau_j \to \infty$ as $j \to \infty$.
\end{proof}

\subsection{Fluid limit}
\label{ch2-sub: dynamic fluid limit}

Next we provide a fluid limit that justifies using fluid systems as an asymptotic approximation for the stochastic system $\bs_n \defeq (\bq_n, \ell_n)$ when $n$ is large. Let $S[0, \infty)$ be the space of c\`adl\`ag functions from $[0, \infty)$ into $\R$, with the Skorohod $J_1$-topology. The stochastic systems $\bs_n$ take values in $D_{\R^\N}[0, \infty) \times S[0, \infty)$, and we endow this space with the product topology. As in Section \ref{ch2-sub: static fluid limit}, we may construct the stochastic systems on a common probability space for all $n$. We adopt such a construction to state the next result, which holds for any initial condition $(q_0, \ell_0)$ such that
\begin{equation*}
	q_0\left(\ell_0\right) > \alpha \quad \text{and} \quad q_0\left(h_0\right) < 1 \quad \text{almost surely},
\end{equation*}
where $h_0 \defeq \ell_0 + 1$. This condition ensures that the limiting threshold is not modified at time zero. Without it we may have sequences of systems where the threshold is modified at the first arrival, or even at each of the first $k$ arrivals for some $k > 1$.

\begin{theorem}
	\label{ch2-the: dynamic fluid limit}
	Suppose that, with probability one, $\bq_n(0) \to q_0$ in the product topology and $\ell_n(0) \to \ell_0$ as $n \to \infty$. Then $\set{\bs_n}{n \geq 1}$ is almost surely relatively compact in $D_{\R^\N}[0, \infty) \times S[0, \infty)$ and the limit of every convergent subsequence is a fluid system.
\end{theorem}

The almost sure relative compactness of $\set{\bq_n}{n \geq 1} \subset D_{\R^\N}[0, \infty)$ can be proved using the methodology of \cite{bramson1998state}, as in Section \ref{ch2-sub: static fluid limit}. If $\set{\bq_{k}}{k \in \calK}$ is a convergent subsequence, then the challenge is to show that the thresholds $\ell_{k}$ converge in $S[0, \infty)$, and to characterize the limits of $\bq_{k}$ and $\ell_{k}$ jointly. In Section \ref{ch2-sap: linear schemes} we do this by induction, approaching $\bs_n$ by systems $\bs_n^m$ where only $m$ threshold updates occur.

\subsection{Convergence of the threshold}
\label{ch2-sub: convergence of the threshold}

In Theorem \ref{ch2-the: dynamic fluid limit}, the limit of a convergent subsequence is a fluid system. Next we prove that the time-dependent threshold of a fluid system eventually reaches an equilibrium value. For this purpose we fix a fluid system $\bs = (\bq, \ell)$ and we consider the associated total mass and tail mass functions, as in Definition \ref{ch2-def: total mass and tail functions}. The next result provides upper bounds for the tail mass functions.

\begin{proposition}
	\label{ch2-prop: asymptotic upper bounds for tail functions in the dynamic case}
	Suppose that there exist $m \geq 0$ and $0 \leq a < b$ such that $\ell(t) \leq m$ if $t \in (a, b)$. Then the following inequalities hold:
	\begin{subequations}
		\label{ch2-eq: upper bound for tail functions}
		\begin{align}
			&\bv_{m + 1}(t) \leq (\lambda - m)^+ + \left[\bv_{m + 1}(a) - (\lambda - m)^+\right] \e^{-(t - a)}, \label{ch2-seq: upper bound for v_i+1} \\
			&\bv_{m + 2}(t) \leq \bv_{m + 2}(a)\e^{-(t - a)}, \label{ch2-seq: upper bound for v_h+2}
		\end{align}
	\end{subequations}
	for all $t \in [a, b]$. If in addition $\bq(t, m) < 1$ for all $t \in (a, b)$, then we have
	\begin{equation}
		\label{ch2-eq: upper bound for v_i+1 when q_l < 1}
		\bv_{m + 1}(t) \leq \bv_{m + 1}(a)\e^{-(t - a)} \quad \text{for all} \quad t \in [a, b].
	\end{equation}
\end{proposition}

\begin{proof}
	Note that $h(t) \leq m + 1$ for all $t \in (a, b)$ by assumption, and $\bq(h) < 1$ almost everywhere by Definition \ref{ch2-def: fluid system}. It follows that
	\begin{equation*}
		\lambda p_{m + 1}(\bs) \leq \left[\lambda - m\left(1 - \bq(m + 1\right)\right]^+ \quad \text{and} \quad p_i(\bs) = 0 \quad \text{for all} \quad  i > m + 1
	\end{equation*}
	almost everywhere in $(a, b)$. The proof of \eqref{ch2-eq: upper bound for tail functions} is as in Proposition \ref{ch2-prop: asymptotic upper bounds for tail functions} and \eqref{ch2-eq: upper bound for v_i+1 when q_l < 1} follows similarly, noting that $\bq(m) < 1$ implies that $p_i(\bs) = 0$ for all $i \geq m + 1$.
\end{proof}

We now prove that the threshold of a fluid system reaches an equilibrium.

\begin{theorem}
	\label{ch2-the: convergence of the threshold}
	There exist $t_{\eq} \geq 0$ and $\ell_{\eq} \in \N$ such that $\ell(t) = \ell_{\eq}$ for all $t \geq t_{\eq}$. 
\end{theorem}

\begin{proof}
	By \eqref{ch2-eq: total mass}, we have $\bu(t) < \flambda + 1$ for all $t \geq t_0$ and some $t_0 \geq 0$. Hence,
	\begin{align}
		\label{ch2-eq: obstruction for threshold increase}
		\tau_j \geq t_0 \quad \text{and} \quad \ell_j \geq \flambda \quad \text{imply} \quad \bq(t, h_j) \leq \frac{\bu(t)}{h_j} < 1 \quad \text{if} \quad t \in [\tau_j, \tau_{j + 1}).
	\end{align}
	This further implies that one of the following two events must occur: the threshold is decreased at $\tau_{j + 1} < \infty$ or no further threshold updates occur and $\tau_{j + 1} = \infty$.
	
	Suppose that $\ell(t) \geq \flambda$ for all $t \geq t_0$. The previous observation implies that $\ell$ is a nonincreasing and lower bounded function within $[t_0, \infty)$. Since $\ell$ is integer-valued, it must eventually settle at some $\ell_{\eq} \geq \flambda$.
	
	Alternatively, assume that there exists $t_1 \geq t_0$ such that $\ell(t_1) < \flambda$. Note that $\ell$ cannot increase beyond $\flambda$ after $t_1$ by \eqref{ch2-eq: obstruction for threshold increase}. Hence, $\ell(t) \leq \floor{\lambda}$ for all $t \geq t_1$. Using Proposition \ref{ch2-prop: asymptotic upper bounds for tail functions in the dynamic case} with $m = \flambda$, $a = t_1$ and $b = \infty$, we obtain
	\begin{align*}
		\bv_{\flambda + 1}(t) \leq \lambda - \flambda + \left[\bv_{\flambda + 1}(t_1) - \left(\lambda - \flambda\right)\right]\e^{-(t - t_1)} \quad \text{for all} \quad t \geq t_1.
	\end{align*}
	The right-hand side converges to $\lambda - \flambda$ as $t \to \infty$ and $\bu(t) \to \lambda$ by \eqref{ch2-eq: total mass}. As a result, there exists $t_2 \geq t_1$ such that
	\begin{equation*}
		\bu(t) - \bv_{\flambda + 1}(t) > \flambda - (1 - \alpha) \quad \text{for all} \quad t \geq t_2.
	\end{equation*}
	
	Suppose that $\tau_j \geq t_2$. Then $\ell_j \leq \flambda$ and therefore
	\begin{equation*}
		\bq(t, \ell_j) \geq \bu(t) - \left(\flambda - 1\right) - \bv_{\flambda + 1}(t) > \alpha \quad \text{for all} \quad t \in \left[\tau_j, \tau_{j + 1}\right).
	\end{equation*}
	This implies that the threshold increases at $\tau_{j + 1} < \infty$ or no further updates occur and $\tau_{j + 1} = \infty$. In the former case, $\ell$ is nondecreasing and upper bounded in $[t_2, \infty)$, and thus eventually reaches an equilibrium value $\ell_{\eq} \leq \flambda$.
\end{proof}

\subsection{Tuning of the learning scheme}
\label{ch2-sub: parameter setting}

Theorem \ref{ch2-the: convergence of the threshold} does not provide any information about the equilibrium threshold value $\ell_{\eq}$ and the time $t_{\eq}$ required for the threshold to reach equilibrium. Particularly, we would like to know how these quantities depend on $\alpha$ to set this parameter in a suitable manner. In this section we study the possible values of $\ell_{\eq}$ and in the next section we investigate the possible values of $t_{\eq}$. The following proposition provides bounds for $\ell_{\eq}$ that will be used to set $\alpha$; the proof is given in Appendix \ref{ch2-app: proofs of various results}.

\begin{proposition}
	\label{ch2-prop: bounds for leq}
	The following properties hold.
	\begin{enumerate}
		\item[(a)] If $\lambda \notin \Zp$ then $\ell_\eq \geq \flambda$, and if $\lambda \in \Zp$, then $\ell_\eq \geq \lambda - 1$.
		
		\item[(b)] We have $\ell_{\eq} \leq \lambda / \alpha$ both for $\lambda \notin \Zp$ and $\lambda \in \Zp$.
	\end{enumerate}
\end{proposition}

Suppose that an upper bound $\lambda_{\max}$ of $\lambda$ is known. We propose to set $\alpha$ such that
\begin{equation}
	\label{ch2-eq: criteria for setting alpha}
	\alpha > \frac{\lambda_{\max}}{\lambda_{\max} + 1}.
\end{equation}
The right-hand side is increasing in $\lambda_{\max}$, which implies that $\alpha > \lambda / (\lambda + 1)$ for all $\lambda \leq \lambda_{\max}$. It follows from Proposition \ref{ch2-prop: bounds for leq} that $\ell_{\eq} \leq \lambda / \alpha < \lambda + 1$ and thus $\ell_{\eq} \leq \clambda$. By (a) of the same proposition, we conclude that
\begin{equation*}
	\flambda \leq \ell_{\eq} \leq \clambda \quad \text{if} \quad \lambda \notin \N \quad \text{and} \quad \lambda - 1 \leq \ell_{\eq} \leq \lambda \quad \text{if} \quad \lambda \in \N.
\end{equation*}
In other words, $\ell_{\eq}$ is lower bounded by a fluid-optimal threshold value and differs from this value at most by one. Since the above inequalities hold as long as $\lambda \leq \lambda_{\max}$, the upper bound $\lambda_{\max}$ can be selected conservatively.

If $\ell_{\eq} \leq \clambda$, then \eqref{ch2-seq: upper bound for v_i+1} with $m = \clambda$, $a = t_{\eq}$ and $b = \infty$ implies that
\begin{equation*}
	\bv_{\clambda + 1}(t) \leq \bv_{\clambda + 1}(t_{\eq})\e^{-(t - t_{\eq})} \quad \text{for all} \quad t \geq t_{\eq}.
\end{equation*}
Hence, after the threshold reaches an equilibrium, the fraction of server pools with more than $\clambda$ tasks decays at least exponentially fast to zero. Although the system may not attain the ideal distribution of the load $q^*$defined in \eqref{ch2-eq: even distribution of the load}, the fraction of server pools with more than $\clambda$ tasks vanishes over time.

If \eqref{ch2-eq: criteria for setting alpha} holds, then $\ell_{\eq}$ is near-optimal. But the equilibrium threshold will in fact be fluid-optimal in many situations. For example, the following corollary gives a sufficient condition for fluid-optimality of the equilibrium threshold; the proof follows directly from Proposition \ref{ch2-prop: bounds for leq}. Note that we cannot use the corollary to set $\alpha$ since the sufficient condition depends on the unknown value of $\lambda$.

\begin{corollary}
	\label{ch2-cor: condition for optimal leq}
	Suppose that
	\begin{equation}
		\label{ch2-eq: optimality condition on alpha}
		\frac{\lambda}{\flambda + 1} < \alpha.
	\end{equation}
	Then $\ell_{\eq} = \flambda$ if $\lambda \notin \Zp$ and $\ell_{\eq} \in \{\lambda - 1, \lambda\}$ if $\lambda \in \Zp$.
\end{corollary}

The corollary says that fluid-optimality of the equilibrium threshold may be lost only when $\lambda$ is close enough to an integer from below. For each $\alpha$ we may find values $\lambda$ that violate \eqref{ch2-eq: criteria for setting alpha}. However, the set of such $\lambda$ decreases to the empty set as $\alpha \to 1$.

\subsection{Convergence time}
\label{ch2-sub: optimality and time until settling}

Assuming that $\lambda \notin \N$ and that the optimality condition \eqref{ch2-eq: optimality condition on alpha} holds, we now focus on the asymptotic time $t_{\eq}$ required by the learning scheme to reach an equilibrium. In particular, the next proposition provides an upper bound $\bar{t}_\eq$ for this time.

\begin{proposition}
	\label{ch2-prop: time until settling}
	Suppose $\lambda \notin \Zp$ and \eqref{ch2-eq: optimality condition on alpha} holds. If
	\begin{equation}
		\label{ch2-eq: upper bound for teq}
		t \geq \bar{t}_{\eq} \defeq \begin{cases}
			\log \left(\frac{\lambda}{\lambda - \flambda}\right) & \text{if} \quad \bu(0) \leq \lambda, \\
			\left[\log \left(\frac{\bu(0) - \lambda}{\alpha\clambda - \lambda}\right)\right]^+ + \log \left(\frac{\lambda}{\lambda - \flambda}\right) & \text{if} \quad \bu(0) > \lambda,
		\end{cases}
	\end{equation}
	then $\ell(t) = \flambda$ and $\bq(t, \flambda) = 1$. In particular, $t_{eq} \leq \bar{t}_{\eq}$.
\end{proposition}

\begin{proof}
	Similarly to \eqref{ch2-eq: obstruction for threshold increase}, we may write
	\begin{equation*}
		\bq(t, \ell_j) \leq \frac{\bu(t)}{\ell_j} = \frac{\lambda + [\bu(0) - \lambda]\e^{-t}}{\ell_j} \quad \text{for all} \quad t \in [\tau_j, \tau_{j + 1}) \quad \text{and} \quad 0 \leq j < \eta,
	\end{equation*}
	where we used \eqref{ch2-eq: total mass} for the last step. We now choose $s_0 \geq 0$ such that the right-hand side is strictly less than $\alpha$ if $\ell_j \geq \clambda$ and $t > s_0$. By \eqref{ch2-eq: optimality condition on alpha}, we may set
	\begin{equation*}
		s_0 = 0 \quad \text{if} \quad \bu(0) \leq \lambda \quad \text{and} \quad s_0 = \left[\log \left(\frac{\bu(0) - \lambda}{\alpha\clambda - \lambda}\right)\right]^+ \quad \text{if} \quad \bu(0) > \lambda.
	\end{equation*}
	Note that $\ell(t) \leq \flambda$ for all $t > s_0$ since $\ell(t) \geq \clambda$ and $t > s_0$ imply that $\bq(t, \ell(t)) < \alpha$, which contradicts Definition \ref{ch2-def: fluid system}. If $\flambda = 0$, then $\ell(t) = \flambda$ for all $t > s_0$ and $\bq(t, \flambda) = \bq(t, 0) = 1$, which completes the proof since $s_0 = \bar{t}_{\eq}$ in this case.
	
	Suppose then that $\flambda > 0$ and let $\bw \defeq \bu - \bv_{\clambda}$. By \eqref{ch2-eq: fluid dynamics revisited},
	\begin{equation*}
		\dot{\bw} = \sum_{i = 1}^\flambda \dot{\bq}(i) \geq \lambda - \flambda > 0 \quad \text{whenever} \quad \ell \leq \flambda \quad \text{and} \quad \bq(\flambda) < 1,
	\end{equation*}
	which implies that $\bw$ is nondecreasing after $s_0$. Indeed, $\ell(t) \leq \flambda$ after $s_0$, so $\bw$ is increasing if $\bq(\flambda) < 1$ and attains its maximum value if $\bq(\flambda) = 1$. Next we prove that there exists $s_1 \in [s_0, \bar{t}_{\eq}]$ such that $\bq(s_1, \flambda) = 1$. Since $\bw$ is nondecreasing within $[s_0, \infty)$, this implies that $\bq(t, \flambda) = 1$ for all $t \geq s_1$ and thus $\ell(t) = \flambda$, which establishes the claim of the proposition.
	
	We argue by contradiction. Suppose that $\bq(t, \flambda) < 1$ for all $t \in [s_0, \bar{t}_\eq]$. Then \eqref{ch2-eq: upper bound for v_i+1 when q_l < 1} holds with $m = \flambda$, $a = s_0$ and $b = \bar{t}_\eq$. This implies that
	\begin{equation*}
		\bw\left(\bar{t}_{\eq}\right) = \bu\left(\bar{t}_{\eq}\right) - \bv_{\clambda}\left(\bar{t}_{\eq}\right) \geq \lambda + \left[\bu(s_0) - \lambda\right]\e^{- \left(\bar{t}_{\eq} - s_0\right)} - \bv_{\clambda}(s_0)\e^{- \left(\bar{t}_{\eq} - s_0\right)} \geq \flambda,
	\end{equation*}
	which contradicts the assumption $\bq\left(\bar{t}_\eq, \flambda\right) < 1$.
\end{proof}

The expressions for $\bar{t}_{\eq}$ provided in \eqref{ch2-eq: upper bound for teq} consist of two terms: the first one upper bounds the time until the threshold falls and remains below $\flambda$ and the second one accounts for the additional amount of time until the threshold reaches $\flambda$ and settles. The first term is zero when $\bu(0) \leq \lambda$. In this case the threshold can never exceed $\flambda$ since $\bu \leq \lambda$; here the expression in \eqref{ch2-eq: upper bound for teq} corresponds to the time required for the total mass function $\bu$ to reach $\flambda$ when $\bu(0) = 0$. Furthermore, \eqref{ch2-eq: upper bound for teq} is tight when $\bu(0) = 0$. If $\bu(0) > \lambda$, then both terms in \eqref{ch2-eq: upper bound for teq} are positive. The first one increases with the initial total mass as one would expect, and more interestingly also depends on $\alpha\clambda - \lambda$. Informally, if the fractional part of $\lambda$ is large, then it might take long for $\bq(\clambda)$ to drop below $\alpha$ so that a threshold update from $\clambda$ to $\flambda$ occurs. The second term is the same as when $\bu(0) \leq \lambda$ and could possibly be reduced.

The following corollary uses the upper bound $\bar{t}_\eq$ to summarize the asymptotic optimality properties of our policy when $\lambda \notin \N$ and \eqref{ch2-eq: optimality condition on alpha} holds. Broadly speaking, the threshold settles at the optimal value $\flambda$ before $\bar{t}_\eq$ in all large enough systems and the occupancy state approaches $q^*$ at least exponentially fast over time. The proof follows directly from Theorem \ref{ch2-the: dynamic fluid limit} and Propositions \ref{ch2-prop: asymptotic upper bounds for tail functions in the dynamic case} and \ref{ch2-prop: time until settling}.

\begin{corollary}
	\label{ch2-the: main optimality result}
	Suppose $\lambda \notin \Zp$ and \eqref{ch2-eq: optimality condition on alpha} holds. There exists $c > 0$ such that
	\begin{align*}
		&\lim_{n \to \infty} \sup_{t \in [\bar{t}_\eq, T]} |\ell_n(t) - \flambda| = 0, \\
		&\lim_{n \to \infty} \sup_{t \in [\bar{t}_\eq, T]} |\bq_n(t, i) - q^*(i)| = 0 \quad \text{for all} \quad i \leq \flambda, \\
		&\limsup_{n \to \infty} \sup_{t \in [\bar{t}_\eq, T]} |\bq_n(t,i) - q^*(i)|\e^{t - \bar{t}_\eq} \leq c \quad \text{for all} \quad i \geq \clambda,
	\end{align*}
	almost surely for all $T \geq \bar{t}_{\eq}$. Also, $c$ may be expressed in terms of $\alpha$, $\lambda$ and $u(0)$.
\end{corollary}

Since the threshold takes values in $\N$, the first limit in the above corollary is equivalent to $\ell_n(t) = \flambda$ for all $t \in [\bar{t}_\eq, T]$ and all sufficiently large $n$.

\section{Simulations}
\label{ch2-sec: simulations}

In this section we analyze the threshold-based dispatching rule and the learning scheme through simulations. First we evaluate whether the threshold indeed reaches an equilibrium value, and the amount of time required for this. Then we analyze the distribution of the number of tasks when the threshold has the optimal value, and we compare with other load balancing policies. Finally, we assess the performance of our threshold-based policy when the arrival rate of tasks is highly variable.

\subsection{Convergence of the threshold}

First we study how large $n$ must be so that $\ell_n$ reaches an equilibrium, as stated in Theorem \ref{ch2-the: convergence of the threshold} for the fluid limit. Figure \ref{ch2-fig: oscillations} shows trajectories of the occupancy and threshold processes for systems with different numbers of server pools. In the system with $n = 100$,  $\ell_n$ oscillates between $\flambda - 1$ and $\clambda$. In the system with $n = 400$, the threshold stays at $\flambda$ most of the time, with sporadic and brief excursions to $\clambda$. The excursions disappear when $n = 500$. This is not shown in Figure \ref{ch2-fig: oscillations}, but can be checked in the other simulations presented in this section.

\begin{figure}
	\centering
	\begin{subfigure}{0.49\columnwidth}
		\centering
		\includegraphics[width = \columnwidth]{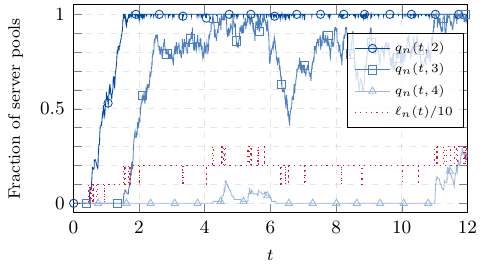}
	\end{subfigure}%
	\hfill
	\begin{subfigure}{0.49\columnwidth}
		\centering
		\includegraphics[width = \columnwidth]{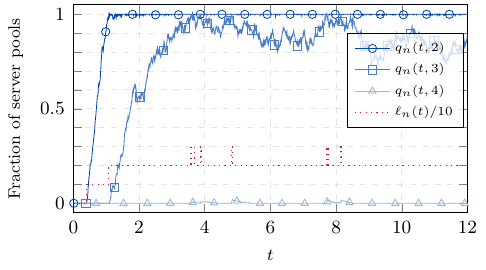}
	\end{subfigure}
	\caption{Evolution of $\ell_n$ over time for $\lambda = 2.9$ and $\alpha = 0.97$. On the left, $n = 100$ and the threshold fluctuates. On the right, $n = 400$ and the threshold stays at $\flambda$ most of the time.}
	\label{ch2-fig: oscillations}
\end{figure}

The convergence of the threshold depends on the fractional part of $\lambda$ besides the number of server pools $n$. If $\ell_n = \flambda$, then Theorem \ref{ch2-the: diffusion limit in the non-integer case} suggests that $\bq_n(\clambda)$ oscillates around $\lambda - \flambda$ with deviations of order $1 / \sqrt{n}$. If $\lambda - \flambda$ is large, then $n$ must also be large so that $\bq_n(\clambda)$ is unlikely to reach one, making the threshold increase. The fractional part of $\lambda$ is relatively large in Figure \ref{ch2-fig: oscillations}, and we see that the threshold is not completely stable at $\flambda$ for $n = 400$.

\subsubsection{Time to reach equilibrium}

We now evaluate the upper bound \eqref{ch2-eq: upper bound for teq} for the time until the threshold settles at the optimal value. Figure \ref{ch2-fig: time until settling} shows trajectories of systems with different initial conditions, where \eqref{ch2-eq: optimality condition on alpha} holds. In both cases $\ell_n$ reaches an equilibrium value quickly, in less than the average amount of time required to execute three tasks. In the initially empty system $\ell_n$ settles at $\flambda$ almost exactly at $\bar{t}_{\eq}$, but in the initially overloaded system the threshold reaches an equilibrium value several units of time before $\bar{t}_\eq$.

\begin{figure}
	\centering
	\begin{subfigure}{0.49\columnwidth}
		\centering
		\includegraphics[width = \columnwidth]{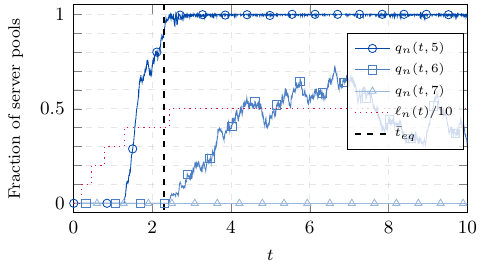}
	\end{subfigure}
	\hfill
	\begin{subfigure}{0.49\columnwidth}
		\centering
		\includegraphics[width = \columnwidth]{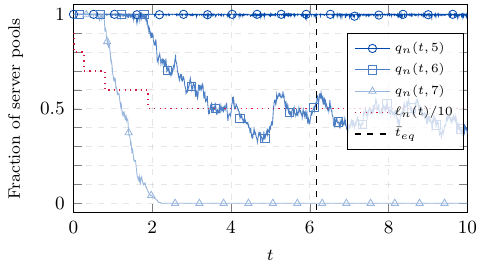}
	\end{subfigure}
	\caption{Time until the threshold reaches an equilibrium for $\lambda = 5.5$, $\alpha = 0.93$ and $n = 500$. The system on the left is initially empty and \eqref{ch2-eq: upper bound for teq} is tight. For the system on the right, all the server pools have $9$ tasks initially and we observe that the upper bound \eqref{ch2-eq: upper bound for teq} is loose.}
	\label{ch2-fig: time until settling}
\end{figure}

\subsection{Distribution of the load}

We now evaluate the distribution of the number of tasks across the server pools in steady state, which has an impact on the quality  of service experienced by users. For this purpose we ran long simulations for several load balancing policies, and we computed the empirical distribution of the fraction of resources received by an arbitrary task. We assumed that the resources of each server pool were equitably distributed among the tasks sharing it, and we computed at each instant of time the number of tasks receiving a certain fraction of resources. We then integrated these quantities over time and normalized them, as shown in Figure \ref{ch2-fig: histogram}.

\begin{figure}
	\centering
	\begin{minipage}{0.49\textwidth}
		\centering
		\includegraphics[width = \columnwidth]{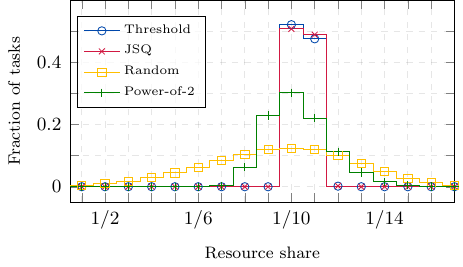}
		\caption{Resources received by tasks under different dispatching rules. All the policies were simulated with $\lambda = 10.5$ and $n = 500$.}
		\label{ch2-fig: histogram}
	\end{minipage}\hfill
	\begin{minipage}{0.49\textwidth}
		\centering
		\includegraphics[width = \columnwidth]{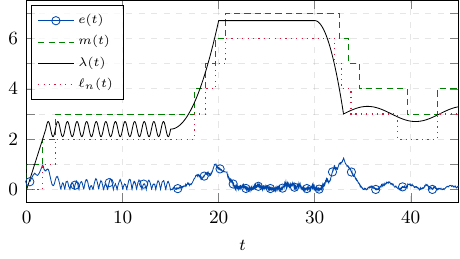}
		\caption{Response of the threshold policy to a highly variable demand, with $n = 500$ and $\alpha = 0.91$ as in \eqref{ch2-eq: criteria for setting alpha} for $\lambda_{\max} = 10$.}
		\label{ch2-fig: time varying}
	\end{minipage}
\end{figure}

The policy that assigns every incoming task to a server pool selected uniformly at random exhibits the largest variance, with some tasks receiving all the resources of a given server pool and some others contending for resources with as many as $15$ tasks. Users are treated more fairly when tasks are sent to the least congested of two server pools selected uniformly at random, but still some tasks share a server pool with as many as $13$ other tasks. Finally, virtually all tasks share a server pool with $10$ or $11$ other tasks when JSQ or the threshold-based policy are used. For these two policies the load is evenly distributed and tasks are treated fairly, so no user experiences an inferior quality of service. The Schur-concave utilities of Section \ref{ch2-sec: model description}, which measure the overall quality of service provided to users, are maximized.

\subsection{Fluctuating demand}

We conclude by studying the response of the learning scheme to highly variable demand patterns. In particular, the trajectories depicted in Figure \ref{ch2-fig: time varying} correspond to a system where $\lambda$ is time-varying and $\alpha$ satisfies \eqref{ch2-eq: criteria for setting alpha}. The system copes effectively with drastic and abrupt load fluctuations, such as the ones at $t = 0, 20, 30$; in all these cases $\ell_n$ quickly reaches the new optimal value. Further, the small but swift load fluctuations in the time interval $[2, 15]$ do not move the threshold away from the optimal value, which remains constant along the entire interval. Finally, $\ell_n$ adjusts to the slow oscillations in $[33, 45]$, which result in changes of the optimal value.

Along with the threshold, we plot
\begin{equation*}
	m(t) \defeq \max\set{i \geq 0}{\bq_n(t, i) > 0} \quad \text{and} \quad e(t) \defeq \norm{\bq_n(t) - q^*(t)}_2, 
\end{equation*}
where $\norm{\scdot}_2$ is the standard Euclidean norm and $q^*(t)$ is computed in terms of $\lambda(t)$ from \eqref{ch2-eq: even distribution of the load}. We note that $m(t)$ equals $\ceil{\lambda(t)}$ most of the time and $e(t)$ is usually small, with peak values coinciding with the most drastic fluctuations of the arrival rate. This means that concentrations of tasks at individual server pools are avoided and the loads are close to balanced most of the time.

\begin{appendices}
	
	\section{Proofs of various results}
	\label{ch2-app: proofs of various results}
	
	\begin{proof}[Proof of Lemma \ref{ch2-lem: interchange of infinite sum and derivative}]
		By \eqref{ch2-eq: static fluid limit} and Tonelli's theorem,
		\begin{align*}
			\bv_j(t) &= \sum_{i = j}^\infty \left[\bq(0, i) + \int_0^t \left(\lambda p_i\left(\bq(s), \ell\right) - i\left[\bq(s, i) - \bq(s, i + 1)\right]\right)ds\right] \\
			&= \bv_j(0) + \sum_{i = j}^\infty \int_0^t \lambda p_i\left(\bq(s), \ell\right)ds - \sum_{i = j}^\infty \int_0^t i\left[\bq(s, i) - \bq(s, i + 1)\right]ds \\
			&= \bv_j(0) + \int_0^t \lambda \sum_{i = j}^\infty p_i\left(\bq(s), \ell\right)ds - \int_0^t \left[(j - 1)\bq(s, j) + \bv_j(s)\right]ds. 
		\end{align*}
		It follows that the differential version holds almost everywhere.
	\end{proof}
	
	\begin{proof}[Proof of Lemma \ref{ch2-lem: exponential bound}]
		The function defined on $[a, b)$ by
		\begin{equation*}
			x \mapsto f(a) + \int_a^x \left[\varphi(y) - f(y)\right]dy - f(x)
		\end{equation*} 
		is nonnegative, absolutely continuous on finite intervals and equal to zero at $x = a$. Thus, there exists a locally integrable nonnegative function $\map{\sigma}{[a, b)}{\R}$ such that
		\begin{equation}
			\label{ch2-aux: integral equation}
			f(x) = f(a) + \int_a^x \left[\varphi(y) - f(y) - \sigma(y)\right]dy.
		\end{equation}
		
		By direct integration, it is possible to check that
		\begin{align*}
			f(x) = f(a)\e^{-(x - a)} + \int_a^x \left[\varphi(y) - \sigma(y)\right]\e^{-(x - y)}dy
		\end{align*}
		solves \eqref{ch2-aux: integral equation}. If $f(a)$, $\varphi$ and $\sigma$ are given, then this solution is unique. Indeed, suppose that $f$ and $g$ solve \eqref{ch2-aux: integral equation}. Then $h \defeq f - g$ satisfies $h(a) = 0$ and $\dot{h} = - h$ almost everywhere in $[a, b)$. We conclude that $h = 0$ and thus $f = g$. 
		
		Since $\sigma$ is nonnegative, we conclude that
		\begin{align*}
			f(x) &= f(a)\e^{-(x - a)} + \int_a^x \left[\varphi(y) - \sigma(y)\right]\e^{-(x - y)}dy \\
			&\leq f(a)\e^{-(x - a)} + \int_a^x \varphi(y)\e^{-(x - y)}dy.
		\end{align*}
		This completes the proof.
	\end{proof}
	
	\begin{proof}[Proof of Proposition \ref{ch2-prop: bounds for leq}]
		Proposition \ref{ch2-prop: asymptotic upper bounds for tail functions in the dynamic case} with $m = \ell_{\eq}$, $a = t_{\eq}$ and $b = \infty$ yields
		\begin{equation*}
			\bv_{h_{\eq} + 1}(t) \leq \bv_{h_\eq + 1}(t_\eq)\e^{-(t - t_\eq)} \quad \text{for all} \quad t \geq t_{\eq} \quad \text{with} \quad h_{\eq} \defeq \ell_{\eq} + 1.
		\end{equation*}
		It follows that $\bu(t) - \bv_{h_{\eq} + 1}(t) \to \lambda$ as $t \to \infty$ by \eqref{ch2-eq: total mass}. As a result,
		\begin{align*}
			h_{\eq} \geq \lim_{t \to \infty} \sum_{i = 1}^{h_{\eq}} \bq(t, i) = \lim_{t \to \infty} \left[\bu(t) - \bv_{h_{\eq} + 1}(t)\right] = \lambda.
		\end{align*}
		This implies that $\ell_{\eq} \geq \lambda - 1$ for all $\lambda > 0$. Moreover, if $\lambda \notin \N$, then $h_{\eq} \geq \clambda$ and thus $\ell_{\eq} \geq \flambda$, which proves (a).
		
		In order to prove (b), note that $\bq(t, i) \geq \bq(t, \ell_{\eq}) \geq \alpha$ for all $t \geq t_{\eq}$ and $i \leq \ell_{\eq}$ by Definition \ref{ch2-def: fluid system}. Therefore,
		\begin{align*}
			0 \leq \limsup_{t \to \infty} \bq(t, h_{\eq}) &= \limsup_{t \to \infty} \left[\bu(t) - \bv_{h_{\eq} + 1}(t) - \sum_{i = 1}^{\ell_{\eq}} \bq(t, i)\right] \leq \lambda - \alpha\ell_{\eq}.
		\end{align*}
		This implies that (b) holds.
	\end{proof}
	
	\section{Proofs of the fluid limits}
	\label{ch2-app: proof of the fluid limits}
	
	Below we prove Theorems \ref{ch2-the: static fluid limit} and \ref{ch2-the: dynamic fluid limit}. First the stochastic systems $\bs_n = (\bq_n, \ell_n)$ are defined on a common probability space in Section \ref{ch2-sap: construction on a common probability space}. In Section \ref{ch2-sap: relative compactness of sample paths} we show that the sequence $\set{\bq_n}{n \geq 1}$ is almost surely relatively compact with respect to a suitable metric. All systems have the same constant threshold in Section \ref{ch2-sap: static threshold}, where we prove Theorem \ref{ch2-the: static fluid limit}. In Section \ref{ch2-sap: linear schemes} the thresholds are adjusted as described in Section \ref{ch2-sec: learning the optimal threshold} and we provide the proof of Theorem \ref{ch2-the: dynamic fluid limit}.
	
	\subsection{Construction on a common probability space}
	\label{ch2-sap: construction on a common probability space}
	
	Consider the following stochastic processes and random variables.
	\begin{itemize}
		\item \emph{Driving Poisson processes:} a Poisson process $\calN_\lambda$ of rate $\lambda$ and a sequence of Poisson processes $\set{\calN_1^i}{i \geq 1}$ of unit rates. These are independent stochastic processes defined on a common probability space $(\Omega_D, \calF_D, \prob_D)$.
		
		\item \emph{Selection variables:} a sequence $\set{U_j}{j \geq 1}$ of independent random variables, uniformly distributed on $[0, 1)$ and defined on a probability space $(\Omega_S, \calF_S, \prob_S)$.
		
		\item \emph{Initial conditions:} random variables $\set{\left(\bq_n(0), \ell_n(0)\right)}{n \geq 1}$ that describe the initial conditions of the systems and are defined on the common probability space $(\Omega_I, \calF_I, \prob_I)$. The random variable $\ell_n(0)$ takes values in $\N$ and $\bq_n(0)$ takes values in the set $\set{q \in Q}{nq(i) \in \N \ \text{for all}\ i \geq 0}$.
	\end{itemize}
	For the last set of random variables, recall that
	\begin{equation*}
		Q \defeq \set{q \in [0, 1]^\N}{q(i + 1) \leq q(i) \leq q(0) = 1\ \text{for all}\ i \geq 1\ \text{and}\ \sum_{i = 1}^\infty q(i) < \infty}.
	\end{equation*}
	
	The processes $\bs_n = (\bq_n, \ell_n)$ will be defined on the completion $(\Omega, \calF, \prob)$ of the product probability space of the spaces $(\Omega_D, \calF_D, \prob_D)$, $(\Omega_S, \calF_S, \prob_S)$ and $(\Omega_I, \calF_I, \prob_I)$  from a set of stochastic equations defined in terms of the above primitives. This type of construction is standard; for example, see \cite{tsitsiklis2013power,stolyar2015pull}.
	
	\subsubsection{Preliminary notation}
	
	Let $q \in Q$ represent the occupancy state of a system. The intervals
	\begin{align*}
		I_i(q) \defeq \left[1 - q(i - 1), 1 - q(i)\right), \quad \text{with} \quad i \geq 1,
	\end{align*}
	form a partition of $[0, 1)$ and have lengths which are proportional to the number of server pools with exactly $i - 1$ tasks. If $q(\ell) < 1$ for some $\ell \in \N$, representing the threshold, then $[0, 1)$ may also be partitioned into the intervals
	\begin{align*}
		J_i(q, \ell) \defeq \left[\frac{1 - q(i - 1)}{1 - q(\ell)}, \frac{1 - q(i)}{1 - q(\ell)}\right) \quad \text{with} \quad 1 \leq i \leq \ell.
	\end{align*}
	Observe that the length of $J_i(q, \ell)$ is the fraction of server pools with precisely $i - 1$ tasks divided by the fraction of server pools with at most $\ell$ tasks; if $q(\ell) = 1$, then we define $J_i(q, \ell) \defeq \emptyset$ for all $1 \leq i \leq \ell$. Letting $h \defeq \ell + 1$, we may now let
	\begin{align}
		\label{ch2-eq: load balancing policy}
		r_{ij}(q, \ell) \defeq \begin{cases}
			\ind{U_j \in J_i\left(q, \ell\right)} & \text{if} \quad i - 1 < \ell, \\
			\ind{q(\ell) = 1}\ind{q(h) < 1} & \text{if} \quad i - 1 = \ell, \\
			\ind{q(h) = 1}\ind{U_j \in I_i\left(q\right)} & \text{if} \quad i - 1 > \ell,
		\end{cases}
		\quad \text{for all} \quad i,j \geq 1.
	\end{align}
	
	We will use the functions $r_{ij}$ to dispatch tasks in the stochastic systems. Namely, if $q$ is the occupancy state and $\ell$ is the threshold, then the $j$th incoming task is dispatched to a server pool with exactly $i - 1$ tasks if and only if $r_{ij}(q, \ell) = 1$. Note that for each fixed $j$ the functions $r_{ij}(q, \ell)$ take values in $\{0, 1\}$ and add up to one.
	
	\subsubsection{Stochastic equations}
	
	We postulate that $\calN_\lambda(nt)$ is the number of tasks that arrive to the system in the interval $[0, t]$. In addition, we let $\set{\sigma_{nj}}{j \geq 1}$ denote the arrival times.
	
	For each pair of functions $\map{\bq}{[0, \infty)}{Q}$ and $\map{\ell}{[0, \infty)}{\N}$, we define two processes: $\calA_n(\bq, \ell)$ and $\calD_n(\bq)$, with values in $\N^\N / n$. We define $\calA_n(\bq, \ell, t, i) \defeq 0$ and $\calD_n(\bq, t, i) \defeq 0$ for $i = 0$ and all $t \geq 0$. For each $i \geq 1$ we let:
	\begin{align*}
		&\calA_n(\bq, \ell, t, i) \defeq \frac{1}{n}\sum_{j = 1}^{\calN_\lambda(n t)} r_{ij}\left(\bq\left(\sigma_{nj}^-\right), \ell\left(\sigma_{nj}^-\right)\right), \\
		&\calD_n(\bq, t, i) \defeq \frac{1}{n}\calN_1^i\left(n\int_0^ti\left[\bq_i(s) - \bq_{i + 1}(s)\right]ds\right).
	\end{align*}
	
	Suppose that the thresholds are adjusted over time according to the control rule described in Section \ref{ch2-sec: learning the optimal threshold}. Then the stochastic system $\bs_n = (\bq_n, \ell_n)$ is defined as the unique solution of the following set of stochastic equations:
	\begin{subequations}
		\label{ch2-eq: implicit equations}
		\begin{align}
			\bq(t) &= \bq_n(0) + \calA_n(\bq, \ell, t) - \calD_n(\bq, t), \label{ch2-seq: aggregate state} \\ 
			\ell(t) &= \ell_n(0) + \sum_{j = 1}^{\calN_\lambda(n t)} \left[\ind{\bq\left(\sigma_{nj}^-, h\left(\sigma_{nj}^-\right)\right) \geq 1 - 1 / n} - \ind{\bq\left(\sigma_{nj}^-, \ell\left(\sigma_{nj}^-\right)\right) \leq \alpha}\right], \label{ch2-seq: threshold}
		\end{align}
	\end{subequations}
	where $h(t) \defeq \ell(t) + 1$ and the unknowns are $\bq$ and $\ell$. If the thresholds are constant over time, then \eqref{ch2-seq: threshold} is replaced by $\ell(t) = k$ for some fixed $k \in \N$ and all $t \geq 0$.
	
	In both cases it is possible to prove by induction on the jump times of the driving Poisson processes that a unique solution defined on $[0, \infty)$ exists almost surely. If we let $\calA_n \defeq \calA_n(\bq_n, \ell_n)$ and $\calD_n \defeq \calD_n(\bq_n, \ell_n)$, then we may interpret \eqref{ch2-eq: implicit equations} as follows.
	\begin{itemize}
		\item The process $\calA_n(i)$ counts the number of arrivals to server pools with exactly $i - 1$ tasks. That it has a jump at an arrival time depends on the dispatching decisions encoded in $r_{ij}\left(\bq_n, \ell_n\right)$. The random variables $\calA_n(t, i)$ add up to $\calN_\lambda(nt)$.
		
		\item The process $\calD_n(i)$ counts the number of departures from server pools with exactly $i$ tasks. It is a Poisson process of rate $n i\left[\bq_n(i) - \bq_n(i + 1)\right]$, equal to the number of tasks in server pools with exactly $i$ tasks.
		
		\item The threshold $\ell_n$ is only adjusted at the arrival times. It increases by one if $n\bq_n(h_n) \geq n - 1$ right before the arrival, and decreases by one if $\bq_n(\ell_n) \leq \alpha$ right before the arrival.
	\end{itemize}
	
	Note that the processes $\calA_n$ and $\calD_n$ have nondecreasing components that are equal to zero at time zero. Also, $\bq_n$ takes values in $Q$; since the total number of tasks in the initial occupancy state is finite, the total number of tasks remains finite.
	
	\subsection{Relative compactness of occupancy processes}
	\label{ch2-sap: relative compactness of sample paths}
	
	In this section we prove that, for $\omega$ in a set of probability one, the sequences
	\begin{equation}
		\label{ch2-aux: sequences}
		\set{\calA_n(\omega)}{n \geq 1}, \quad \set{\calD_n(\omega)}{n \geq 1} \quad \text{and} \quad \set{\bq_n(\omega)}{n \geq 1}
	\end{equation}
	are relatively compact in $D_{\R^\N}[0, \infty)$; i.e., their closures are compact. As we explain below, $D_{\R^\N}[0, \infty)$ is a metrizable space and thus relative compactness is equivalent to sequential compactness. In particular, we prove that every subsequence of the above sequences has a further subsequence that converges in $D_{\R^\N}[0, \infty)$.
	
	Consider the metric
	\begin{equation*}
		d(x, y) \defeq \sum_{i = 0}^\infty \frac{\min\{|x(i) - y(i)|, 1\}}{2^i} \quad \text{for all} \quad x, y \in \R^\N,
	\end{equation*}
	which induces the product topology in $\R^\N$. Let $D_{\R^\N}[0, T]$ denote the space of c\`adl\`ag functions from $[0, T]$ into $\R^\N$. We equip this space with the uniform metric:
	\begin{equation*}
		\rho_u^T(\bx, \by) \defeq \sup_{t \in [0, T]} d(\bx(t), \by(t)) \quad \text{for all} \quad \bx, \by \in D_{\R^\N}[0, T].
	\end{equation*}
	The space $D_{\R^\N}[0, \infty)$ is metrizable since the topology of uniform convergence over compact sets is compatible with the metric defined by
	\begin{equation*}
		\rho_u^\infty \defeq \sum_{T = 0}^\infty \frac{\min\{\rho_u^T(\bx, \by), 1\}}{2^T} \quad \text{for all} \quad \bx, \by \in D_{\R^\N}[0, \infty).
	\end{equation*}
	
	We first prove that the sequences in \eqref{ch2-aux: sequences} are relatively compact in $D_{\R^\N}[0, T]$ for all $T \geq 0$ and all $\omega$ in a set of probability one; here we are actually referring to the restrictions to $[0, T]$ of the functions $\calA_n(\omega)$, $\calD_n(\omega)$ and $\bq_n(\omega)$. Then we establish the relative compactness in $D_{\R^\N}[0, \infty)$, also almost surely.
	
	As in the statements of Theorems \ref{ch2-the: static fluid limit} and \ref{ch2-the: dynamic fluid limit}, we assume throughout this section that there exists a random variable $q_0$ with values in $Q$ such that
	\begin{equation}
		\label{ch2-ass: convergence of initial conditions}
		\lim_{n \to \infty} d\left(\bq_n(0), q_0\right) = 0 \quad \text{almost surely.}
	\end{equation}
	The subsequent arguments are based on \cite{bramson1998state}; see \cite[Section 3.3.2]{zubeldia2019delay} as well. These arguments make no use of the specific dynamics of the threshold. In particular, the results hold both when the threshold is constant and when it is adjusted over time.
	
	\begin{proposition}
		\label{ch2-prop: set of probability one}
		There exists a set of probability one $\Gamma_T$ where:
		\begin{subequations}
			\label{ch2-eq: convergence of initial condition and strong laws for stochastic primitives}
			\begin{align}
				&\lim_{n \to \infty} d\left(\bq_n(0), q_0\right) = 0, \label{ch2-seq: convergence of initial condition}\\
				&\lim_{n \to \infty} \sup_{t \in [0, T]} \left|\frac{1}{n}\calN_\lambda(n t) - \lambda t\right| = 0, \label{ch2-seq: law of large numbers for arrivals}\\
				&\lim_{n \to \infty} \sup_{t \in [0, iT]} \left|\frac{1}{n}\calN_1^i(n t) - t\right| = 0 \quad \text{for all} \quad i \geq 1. \label{ch2-seq: law of large numbers for departures} \\
				&\lim_{k \to \infty} \frac{1}{k}\sum_{j = 1}^{k} \ind{U_j \in [a, b)} = b - a \quad \text{for all} \quad [a, b) \subset [0, 1). \label{ch2-seq: glivenko-cantelli}
			\end{align}
		\end{subequations}
	\end{proposition}
	
	\begin{proof}
		This result is a straightforward consequence of the functional strong law of large numbers for the Poisson process, which applies to \eqref{ch2-seq: law of large numbers for arrivals} and \eqref{ch2-seq: law of large numbers for departures}, and the strong law of large numbers for independent and identically distributed random variables, which applies to \eqref{ch2-seq: glivenko-cantelli}.
	\end{proof}
	
	\begin{remark}
		We will use \eqref{ch2-seq: convergence of initial condition}, \eqref{ch2-seq: law of large numbers for arrivals} and \eqref{ch2-seq: law of large numbers for departures} in this section, but \eqref{ch2-seq: glivenko-cantelli} will only be used later, to characterize the limits of convergent subsequences.
	\end{remark}
	
	By \eqref{ch2-seq: aggregate state} and \eqref{ch2-seq: convergence of initial condition}, it suffices to prove that $\calA_n(\omega)$ and $\calD_n(\omega)$ form relatively compact sequences for all $\omega \in \Gamma_T$. Consider the space $D[0, T]$ of all c\`adl\`ag functions from $[0, T]$ into $\R$ with the uniform norm:
	\begin{align*}
		\norm{\bx}_T = \sup_{t \in [0, T]} |\bx(t)| \quad \text{for all} \quad \bx \in D[0, T].
	\end{align*}
	Note that $\bx_n \to \bx$ as $n \to \infty$ in $D_{\R^\N}[0, T]$ if and only if $\bx_n(i) \to \bx(i)$ as $n \to \infty$ in $D[0, T]$ for all $i \geq 0$. Furthermore, we have the following proposition.
	
	\begin{proposition}
		\label{ch2-prop: relative compactness of coordinates}
		The sequence $\set{\bx_n}{n \geq 1}$ is relatively compact in $D_{\R^\N}[0, T]$ if and only if $\set{\bx_n(i)}{n \geq 1}$ is relatively compact in $D[0, T]$ for all $i \geq 0$.
	\end{proposition}
	
	\begin{proof}
		We only need to prove the converse, so assume that  $\set{\bx_n(i)}{n \geq 1}$ is relatively compact for all $i \geq 0$. Given an increasing sequence $\calK \subset \N$, we must show that there exists a subsequence of $\set{\bx_k}{k \in \calK}$ that converges in $D_{\R^\N}[0, T]$. For this purpose, we may define a family of increasing sequences $\set{\calJ_i}{i \geq 0}$ such that:
		\begin{enumerate}
			\item[(a)] $\calJ_{i + 1} \subset \calJ_i \subset \calK$ for all $i \geq 0$,
			
			\item[(b)] $\set{\bx_j(i)}{j \in \calJ_i}$ has a limit $\bx(i) \in D[0, T]$ for each $i \geq 0$.
		\end{enumerate}
		
		Define $\set{k_j}{j \geq 1} \subset \calK$ such that $k_j$ is the $j$th element of $\calJ_j$. Then
		\begin{equation*}
			\lim_{j \to \infty} \norm{\bx_{k_j}(i) - \bx(i)}_T = 0 \quad \text{for all} \quad i \geq 0.
		\end{equation*} 
		Let $\bx \in D_{\R^\N}[0, T]$ be the function with components the functions $\bx(i)$. Then $\bx_{k_j}$ converges to $\bx$ as $j \to \infty$ in $D_{\R^\N}[0, T]$.
	\end{proof}
	
	Let us fix some $\omega \in \Gamma_T$, which we omit from the notation for brevity. As a result of the proposition, it suffices to establish that $\set{\calA_n(i)}{n \geq 1}$ and $\set{\calD_n(i)}{n \geq 1}$ are relatively compact in $D[0, T]$ for all $i \geq 0$. Consider the sets
	\begin{equation*}
		L_M \defeq \set{\bx \in D[0, T]}{\bx(0) = 0\ \text{and}\ \left|\bx(t) - \bx(s)\right| \leq M|t - s|\ \text{for all}\ s, t \in [0, T]},
	\end{equation*}
	which are compact for each $M > 0$ by the Arzel\'a-Ascoli theorem. For each $i \geq 0$, we prove that there exists $M_i$ such that $\calA_n(i)$ and $\calD_n(i)$ approach $L_{M_i}$ as $n$ grows large. Then we use the compactness of $L_{M_i}$ to show that $\set{\calA_n(i)}{n \geq 1}$ and $\set{\calD_n(i)}{n \geq 1}$ are relatively compact subsets of $D[0, T]$. For this purpose, we define
	\begin{equation*}
		L_M^\varepsilon \defeq \set{\bx \in D[0, T]}{\bx(0) = 0\ \text{and}\ \left|\bx(t) - \bx(s)\right| \leq M|t - s| + \varepsilon\ \text{for all}\ s, t \in [0, T]}.
	\end{equation*}
	
	\begin{lemma}
		\label{ch2-lem: bramson}
		If $\bx \in L_M^\varepsilon$, then there exists $\by \in L_M$ such that $\norm{\bx - \by}_T \leq 4\varepsilon$.
	\end{lemma}
	
	The above lemma is a restatement of \cite[Lemma 4.2]{bramson1998state}. Together with the next lemma, it implies that for each $i \geq 0$ there exists $M_i$ such that $\calA_n(i)$ and $\calD_n(i)$ approach the set $L_{M_i}$ of Lipschitz functions of modulus $M_i$ as $n \to \infty$. Recall that we fixed some $\omega \in \Gamma_T$. The following lemma applies any $\omega \in \Gamma_T$.
	
	\begin{lemma}
		\label{ch2-lem: almost lipschitz}
		For each $i \geq 0$ there exist constants $M_i > 0$ and $\set{\varepsilon_n(i) > 0}{n \geq 1}$ such that $\calA_n(i), \calD_n(i) \in L_{M_i}^{\varepsilon_n(i)}$ for all $n$ and $\varepsilon_n(i) \to 0$ as $n \to \infty$.
	\end{lemma}
	
	\begin{proof}
		Since $\calA_n(i)$ and $\calD_n(i)$ are identically zero for $i = 0$, we may focus on the case where $i \geq 1$. For all $s, t \in [0, T]$ we have
		\begin{align*}
			\left|\calA_n(t, i) - \calA_n(s, i)\right| \leq \frac{1}{n} \left|\calN_\lambda(n t) - \calN_\lambda(n s)\right| \leq \lambda|t - s| + 2\sup_{u \in [0, T]} \left|\frac{1}{n}\calN_\lambda(n u) - \lambda u\right|.
		\end{align*}
		By \eqref{ch2-seq: law of large numbers for arrivals}, there exist $\set{\varepsilon_n^1(i) > 0}{n \geq 1}$ such that
		\begin{align*}
			\left|\calA_n(t, i) - \calA_n(s, i)\right| \leq \lambda|t - s| + \varepsilon_n^1(i) \quad \text{for all} \quad s, t \in [0, T] \quad \text{and} \quad \lim_{n \to \infty} \varepsilon_n^1(i) = 0.
		\end{align*}
		
		For each $t \in [0, T]$ let
		\begin{align*}
			\boldf_n(t, i) \defeq \int_0^t i\left[\bq_n(s, i) - \bq_n(s, i + 1)\right]ds.
		\end{align*}
		This function has the following two properties:
		\begin{align*}
			\boldf_n(T, i) \leq iT \quad \text{and} \quad \left|\boldf_n(t, i) - \boldf_n(s, i)\right| \leq i|t - s| \quad \text{for all} \quad s, t \in [0, T].
		\end{align*}
		We conclude that
		\begin{align*}
			\left|\calD_n(t, i) - \calD_n(s, i)\right| &= \frac{1}{n} \left|\calN_1^i\left(n \boldf_n(t, i)\right) - \calN_1^i\left(n \boldf_n(s, i)\right)\right| \\
			&\leq |\boldf_n(t, i) - \boldf_n(s, i)| + 2\sup_{u \in [0, T]} \left|\frac{1}{n}\calN_1^i\left(n \boldf_n(u, i)\right) - \boldf_n(u, i)\right| \\
			&\leq i|t - s| + 2\sup_{u \in [0, iT]} \left|\frac{1}{n}\calN_1^i(n u) - u\right| \quad \text{for all} \quad s, t \in [0, T].
		\end{align*}
		By \eqref{ch2-seq: law of large numbers for departures}, there exist $\set{\varepsilon_{n}^2(i) > 0}{n \geq 1}$ such that
		\begin{align*}
			\left|\calD_n(t, i) - \calD_n(s, i)\right| \leq i|t - s| + \varepsilon_{n}^2(i) \quad \text{for all} \quad s, t \in [0, T] \quad \text{and} \quad \lim_{n \to \infty} \varepsilon_{n}^2(i) = 0.
		\end{align*}
		
		The result follows letting $M_i \defeq \max\left\{\lambda, i\right\}$ and $\varepsilon_n(i) \defeq \max\left\{\varepsilon_n^1(i), \varepsilon_{n}^2(i)\right\}$.
	\end{proof}
	
	We now prove that the sequences in \eqref{ch2-aux: sequences} are relatively compact in $D_{\R^\N}[0, T]$ for all $T \geq 0$ and all $\omega$ in the set of probability one $\Gamma_T$.
	
	\begin{proposition}
		\label{ch2-prop: tightness}
		The restrictions to $[0, T]$ of $\calA_n(\omega)$, $\calD_n(\omega)$ and $\bq_n(\omega)$ constitute relatively compact sequences in $D_{\R^\N}[0, T]$ for all $T \geq 0$ and $\omega \in \Gamma_T$. Moreover, the limit of each convergent subsequence has Lipschitz components.
	\end{proposition}
	
	\begin{proof}
		Fix $T \geq 0$ and $\omega \in \Gamma_T$; we omit $\omega$ from the notation for brevity. It suffices to show that for each $i \geq 1$ every subsequence of $\set{\calA_n(i)}{n \geq 1}$ and $\set{\calD_n(i)}{n \geq 1}$  has a further subsequence that converges in $D[0, T]$ to a Lipschitz function.
		
		The above properties hold for $i = 0$. We now fix $i \geq 1$ and prove these properties for $\set{\calA_n(i)}{n \geq 1}$; the same arguments apply to $\set{\calD_n(i)}{n \geq 1}$.
		
		Let $M_i > 0$ and $\set{\varepsilon_{n}(i) > 0}{n \geq 1}$ be as in the statement of Lemma \ref{ch2-lem: almost lipschitz}. It follows from Lemma \ref{ch2-lem: bramson} that for each $n$ there exists $\ba_n(i) \in L_{M_i}$ such that
		\begin{equation*}
			\norm{\calA_n(i) - \ba_n(i)}_T \leq 4\varepsilon_n(i).
		\end{equation*}
		Recall that $L_{M_i}$ is compact, thus every increasing sequence of natural numbers has a subsequence $\calK$ such that $\set{\ba_k(i)}{k \in \calK}$ converges to a function $\ba \in L_{M_i}$. Moreover,
		\begin{equation*}
			\limsup_{k \to \infty} \norm{\calA_k(i) - \ba(i)}_T \leq \lim_{k \to \infty} \left[4\varepsilon_k(i) + \norm{\ba_k(i) - \ba(i)}_T\right] = 0,
		\end{equation*}
		where the limits are taken along $\calK$. Hence, every subsequence of $\set{\calA_n(i)}{n \geq 1}$ has a further subsequence that converges to a Lipschitz function.
	\end{proof}
	
	The fact that $\set{\calA_n}{n \geq 1}$, $\set{\calD_n}{n \geq 1}$ and $\set{\bq_n}{n \geq 1}$ are relatively compact in $D_{\R^\N}[0, \infty)$ with probability one follows as a corollary.
	
	\begin{theorem}
		\label{ch2-the: tightness}
		For all $\omega$ in a set of probability one $\Gamma_\infty$, the sequences
		\begin{equation*}
			\set{\calA_n(\omega)}{n \geq 1}, \quad \set{\calD_n(\omega)}{n \geq 1} \quad \text{and} \quad \set{\bq_n(\omega)}{n \geq 1}
		\end{equation*}
		are relatively compact in $D_{\R^\N}[0, \infty)$. Also, the limit of every convergent subsequence is a function with locally Lipschitz coordinate functions.
	\end{theorem}
	
	\begin{proof}
		Since $\Gamma_T$ has probability one for all $T \geq 0$, the set
		\begin{equation*}
			\Gamma_\infty \defeq \bigcap_{T \in \N} \Gamma_T
		\end{equation*}
		has probability one as well. We fix some $\omega \in \Gamma_\infty$, which we omit from the notation for brevity. Next we prove that $\set{\bq_n}{n \geq 1}$ is relatively compact in $D_{\R^\N}[0, \infty)$ and such that the limit of every convergent subsequence has locally Lipschitz components. Exactly the same arguments apply if $\bq_n$ is replaced by $\calA_n$ or $\calD_n$.
		
		Fix an arbitrary increasing sequenc $\calK \subset \N$. We must prove that $\set{\bq_k}{k \in \calK}$ has a subsequence that converges uniformly over compact sets to a function with locally Lipschitz components. For this purpose we may construct sequences $\set{\calK_T}{T \in \N}$ with the following two properties.
		\begin{enumerate}
			\item[(a)] $\calK_{T + 1} \subset \calK_T \subset \calK$ for all $T \in \N$.
			
			\item[(b)] For each $T \in \N$, there exists $\bq_T \in D_{\R^\N}[0, T]$ such that $\rho_u^T\left(\bq_k|_{[0, T]}, \bq_T\right) \to 0$ as $k \to \infty$ with $k \in \calK_T$, and the components of $\bq_T$ are Lipschitz.
		\end{enumerate}
		
		Let $k_l$ denote the $l$th element of $\calK_l$. It follows from (a) and (b) that
		\begin{equation}
			\label{ch2-aux: from finite intervals to half line}
			\lim_{l \to \infty} \rho_u^T \left(\bq_{k_l}|_{[0, T]}, \bq_T\right) = 0 \quad \text{for all} \quad T \in \N.
		\end{equation}
		Note that $\bq_S(t) = \bq_T(t)$ for all $t \leq S, T$ since $\bq_{k_l}(t) \to \bq_T(t)$ as $l \to \infty$ if $t \leq T$. We may thus define $\bq \in D_{\R^\N}[0, \infty)$ such that $\bq(t) = \bq_T(t)$ if $t \leq T$. Also, (b) implies that $\bq$ has locally Lipschitz components and \eqref{ch2-aux: from finite intervals to half line} says that $\bq_{k_l} \to \bq$ uniformly over compact sets as $l \to \infty$.
	\end{proof}
	
	\begin{remark}
		\label{ch2-rem: relative compactness independent of threshold dynamics}
		As noted earlier, the proofs of the results stated in this section made no use of the specific dynamics of the threshold. In particular, the previous theorem holds regardless of how the threshold evolves over time.
	\end{remark}
	
	\subsection{Systems with a static threshold}
	\label{ch2-sap: static threshold}
	
	In this section we consider systems where the threshold remains constant and we prove Theorem \ref{ch2-the: static fluid limit}. Specifically, we assume that there exists $\ell \in \N$ such that
	\begin{equation*}
		\ell_n(\omega, t) = \ell \quad \text{for all} \quad n \geq 1, \quad \omega \in \Omega \quad \text{and} \quad t \geq 0.
	\end{equation*}
	We have already proved in Theorem \ref{ch2-the: tightness} that there exists a set $\Gamma_\infty$ of probability one with the following property. If $\omega \in \Gamma_\infty$, then every subsequence of $\set{\bq_n(\omega)}{n \geq 1}$ has a further subsequence that converges uniformly over compact sets. It remains to show that every subsequential limit $\bq$ is such that $\bq(t) \in Q$ for all $t \geq 0$ and satisfies the system of differential equations defined in \eqref{ch2-eq: static fluid limit}. 
	
	For this purpose, we fix an arbitrary $\omega \in \Gamma_\infty$ and some increasing sequence of $\calK \subset \N$ such that $\set{\bq_k(\omega)}{k \in \calK}$ converges uniformly over compact sets; in the sequel we omit $\omega$ from the notation for brevity. By Theorem \ref{ch2-the: tightness}, we may assume without any loss of generality that $\set{\calA_k}{k \in \calK}$ and $\set{\calD_k}{k \in \calK}$ converge uniformly over compact sets to functions $\ba$ and $\bd$, respectively; this may require to replace $\calK$ by a further subsequence, which does not modify the subsequent arguments.
	
	It follows from \eqref{ch2-seq: aggregate state} and \eqref{ch2-seq: convergence of initial condition} that the limit $\bq \in D_{\R^\N}[0, \infty)$ of $\bq_k$ satisfies
	\begin{align*}
		\bq(t) = q_0 + \ba(t) - \bd(t) \quad \text{for all} \quad t \geq 0.
	\end{align*}
	By Theorem \ref{ch2-the: tightness}, the coordinate functions $\ba(i)$ and $\bd(i)$ are locally Lipschitz, thus almost everywhere differentiable. Furthermore, these functions are nondecreasing and satisfy $\ba(0, i) = \bd(0, i) = 0$ since $\calA_k(i)$ and $\calD_k(i)$ have these properties for all $k \in \calK$.
	
	\begin{lemma}
		\label{ch2-lem: derivatives}
		There exists a set $\calR \subset (0, \infty)$ such that the complement of $\calR$ has zero Lebesgue measure and the coordinate functions $\ba(i)$ and $\bd(i)$ are differentiable at every point of $\calR$ for all $i \geq 0$. Furthermore,
		\begin{equation}
			\label{ch2-aux: derivative of d}
			\dot{\bd}(t_0, i) = i\left[\bq(t_0, i) - \bq(t_0, i + 1)\right] \quad \text{for all} \quad i \geq 1 \quad \text{and} \quad t_0 \in \calR,
		\end{equation}
		and the derivatives $\dot{\ba}(t_0, i)$ are as follows.
		\begin{enumerate}
			\item[(a)] If $\bq(t_0, \ell) < 1$, then
			\begin{equation*}
				\dot{\ba}(t_0, i) = \begin{cases}
					\lambda\left[\frac{\bq(t_0, i - 1) - \bq(t_0, i)}{1 - \bq(t_0, \ell)}\right] & \text{if} \quad 1 \leq i \leq \ell, \\
					0 & \text{if} \quad i \geq h.
				\end{cases}
			\end{equation*}
			
			\item[(b)] If $\bq(t_0, \ell) = 1$ and $\bq(t_0, h) < 1$, then
			\begin{equation*}
				\dot{\ba}(t_0, i) = \begin{cases}
					\ell\left[1 - \bq(t_0, h)\right] & \text{if} \quad i = \ell, \\
					\lambda - \ell\left[1 - \bq(t_0, h)\right] & \text{if} \quad i = h, \\
					0 & \text{if} \quad i \neq \ell, h.
				\end{cases}
			\end{equation*}
			
			\item[(c)] If $\bq(t_0, h) = 1$, then
			\begin{equation*}
				\dot{\ba}(t_0, i) = \begin{cases}
					h\left[1 - \bq(t_0, h + 1)\right] & \text{if} \quad i = h, \\
					\left[\lambda - h\left(1 - \bq(t_0, h + 1)\right)\right]\left[\bq(t_0, i - 1) - \bq(t_0, i)\right] & \text{if} \quad i \geq h + 1, \\
					0 & \text{if} \quad 1 \leq i \leq \ell.
				\end{cases}
			\end{equation*}
		\end{enumerate}
	\end{lemma}
	
	\begin{proof}
		The existence of $\calR$ follows from the almost everywhere differentiability of the coordinate functions $\ba(i)$ and $\bd(i)$, which was already noted above.
		
		Fix $t_0 \in \calR$. Note that $\calD_k(i) \to \bd(i)$ and $\bq_k(i) \to \bq(i)$ uniformly over compact sets as $k \to \infty$ for all $i \geq 0$. It follows from these limits and \eqref{ch2-seq: law of large numbers for departures} that 
		\begin{align}
			\label{ch2-eq: explicit formula for d}
			\bd(t, i) = \int_0^t i\left[\bq(s, i) - \bq(s, i + 1)\right]ds \quad \text{for all} \quad i \geq 1 \quad \text{and} \quad t \geq 0.
		\end{align}
		In particular, we conclude that \eqref{ch2-aux: derivative of d} holds.
		
		In order to prove (a), fix $1 \leq i \leq \ell$ and $\varepsilon > 0$. The coordinate functions $\bq(i)$ are continuous, in fact locally Lipzchits, and we have $\bq_k(i) \to \bq(i)$ uniformly over compact sets as $k \to \infty$. Since $\bq(t_0, \ell) < 1$, this implies that
		\begin{align*}
			\frac{1 - \bq_k(t, i - 1)}{1 - \bq_k(t, \ell)} \geq \frac{1 - \bq(t_0, i - 1)}{1 - \bq(t_0, \ell)} - \varepsilon \quad \text{and} \quad \frac{1 - \bq_k(t, i)}{1 - \bq_k(t, \ell)} \leq \frac{1 - \bq(t_0, i)}{1 - \bq(t_0, \ell)} + \varepsilon
		\end{align*}
		for a sufficiently small $\delta > 0$, all $t \in (t_0 - \delta, t_0 + \delta)$ and all large enough $k \in \calK$. For all $k$ and $t$ satisfying the latter conditions, we have
		\begin{align*}
			\tilde{J}_i = \left[\frac{1 - \bq(t_0, i - 1)}{1 - \bq(t_0, \ell)} - \varepsilon, \frac{1 - \bq(t_0, i)}{1 - \bq(t_0, \ell)} + \varepsilon\right) \supset J_i\left(\bq_k(t), \ell\right).
		\end{align*}
		Using the definition of $\calA_k(i)$, we conclude that the following inequality holds for for all $t \in (t_0 - \delta, t_0 + \delta)$ and all sufficiently large $k \in \calK$:
		\begin{align*}
			\calA_k(t, i) - \calA_k(t_0, i) &= \frac{1}{k}\sum_{j = \calN_\lambda(kt_0) + 1}^{\calN_\lambda(kt)} \ind{U_j \in J_i\left(\bq_k\left(\sigma_{kj}^-\right), \ell\right)} \\
			&\leq \frac{1}{k}\sum_{j = \calN_\lambda(kt_0) + 1}^{\calN_\lambda(kt)} \ind{U_j \in \tilde{J}_i} \\
			&= \frac{\calN_\lambda(kt)}{k\calN_\lambda(kt)}\sum_{j = 1}^{\calN_\lambda(kt)} \ind{U_j \in \tilde{J}_i} - \frac{\calN_\lambda(kt_0)}{k\calN_\lambda(kt_0)}\sum_{j = 1}^{\calN_\lambda(kt_0)} \ind{U_j \in \tilde{J}_i}.
		\end{align*}
		
		By taking the limit as $k \to \infty$ on both sides, we obtain
		\begin{align}
			\label{ch2-eq: computing the derivative of a}
			\ba(t, i) - \ba(t_0, i) \leq \lambda(t - t_0)\left[\frac{\bq(t_0, i - 1) - \bq(t_0, i)}{1 - \bq(t_0, \ell)} + 2\varepsilon\right].
		\end{align}
		Here the limit of the right-hand side uses \eqref{ch2-seq: law of large numbers for arrivals} and \eqref{ch2-seq: glivenko-cantelli}. Note that the above inequality is preserved if $t \in (t_0, t_0 + \delta)$ and we divide both sides by $t - t_0$. If we then take the limit as $t \to t_0^+$, and next we take the limit as $\varepsilon \to 0$, then we obtain
		\begin{align*}
			\dot{\ba}(t_0, i) \leq \lambda\left[\frac{\bq(t_0, i - 1) - \bq(t_0, i)}{1 - \bq(t_0, \ell)}\right]. 
		\end{align*}
		If $t \in (t_0 - \delta, t_0)$, then \eqref{ch2-eq: computing the derivative of a} is reversed when we divide by $t - t_0$. Taking the limit as $t \to t_0^-$ and then as $\varepsilon \to 0$, we see that equality holds above, proving (a).
		
		In order to prove (b), note that there exists $\varepsilon > 0$ such that $\bq(t, h) < 1$ for all $t \in (t_0 - \varepsilon, t_0 + \varepsilon)$. Since $\bq_k(h) \to \bq(h)$ uniformly over compact sets as $k \to \infty$, the latter property also holds for $\bq_k(h)$ if $k \in \calK$ is sufficiently large. Hence, tasks arriving to large enough systems during the interval $(t_0 - \varepsilon, t_0 + \varepsilon)$ are exclusively sent to server pools with at most $\ell$ tasks. As a result, we have
		\begin{subequations}
			\begin{align}
				&\sum_{i = 1}^h \left[\calA_k(t, i) - \calA_k(t_0, i)\right] = \frac{1}{k}\left[\calN_\lambda\left(k t\right) - \calN_\lambda\left(k t_0\right)\right], \label{ch2-seq: sum of pre-limits} \\
				&\calA_k(t, i) = \calA_k(t_0, i) \quad \text{if} \quad i \geq h + 1, \label{ch2-seq: pre-limits that are zero}
			\end{align}
		\end{subequations}
		for all $t \in (t_0 - \varepsilon, t_0 + \varepsilon)$ and all sufficiently large $k \in \calK$. The right-hand side of the first equation converges uniformly over $(t_0 - \varepsilon, t_0 + \varepsilon)$ to $\lambda (t - t_0)$ by \eqref{ch2-seq: law of large numbers for arrivals}. By taking the limit as $k \to \infty$ on both sides of \eqref{ch2-seq: sum of pre-limits}, dividing both sides by $t - t_0$ next, and then taking the limit as $t \to t_0$, we conclude that
		\begin{align}
			\label{ch2-eq: sum of derivatives}
			\sum_{i = 1}^h \dot{\ba}(t_0, i) = \lambda.
		\end{align}
		Furthermore, it follows from \eqref{ch2-seq: pre-limits that are zero} that
		\begin{align}
			\label{ch2-eq: zero derivatives}
			\dot{\ba}(t_0, i) = 0 \quad \text{for all} \quad i \geq h + 1.
		\end{align}
		
		Note that $\bq(t_0, i) = 1$ implies that
		\begin{align*}
			0 \leq \lim_{t \to t_0^-} \frac{\bq(t, i) - \bq(t_0, i)}{t - t_0} = \dot{\bq}(t_0, i) = \lim_{t \to t_0^+} \frac{\bq(t, i) - \bq(t_0, i)}{t - t_0} \leq 0;
		\end{align*}
		since $t_0 > 0$ by assumption, we may take both left and right limits. Thus, $\bq(t_0, \ell) = 1$ implies $\dot{\bq}(t_0, i) = 0$ for all $1 \leq i \leq \ell$. We conclude that
		\begin{equation*}
			\dot{\ba}(t_0, i) = \dot{\bd}(t_0, i) = \begin{cases}
				0 & \text{if} \quad 1 \leq i \leq \ell - 1, \\
				\ell\left[1 - \bq(t_0, h)\right] & \text{if} \quad i = \ell.
			\end{cases}
		\end{equation*}
		Now (b) follows directly from \eqref{ch2-eq: sum of derivatives} and \eqref{ch2-eq: zero derivatives}.
		
		Finally, suppose that $\bq(t_0, h) = 1$ and let us prove (c). As above,
		\begin{align}
			\label{ch2-eq: derivatives for i <= h when qh = 1}
			\dot{\ba}(t_0, i) = \dot{\bd}(t_0, i) =  \begin{cases}
				0 & \text{if} \quad 1 \leq i \leq \ell, \\
				h\left[1 - \bq(t_0, h + 1)\right] & \text{if} \quad i = h.
			\end{cases}
		\end{align}
		In order to compute the other derivatives, recall that $\set{\sigma_{kj}}{j \geq 1}$ are the jump times of the arrival process $t \to \calN_\lambda(kt)$, and consider the following process:
		\begin{align*}
			\calC_k(t) &\defeq \frac{1}{k}\sum_{j = \calN_\lambda(kt_0) + 1}^{\calN_\lambda(k t)} \ind{\bq_k\left(\sigma_{kj}^-, h\right) = 1} \\
			&= \frac{1}{k}\left[\calN_\lambda(k t) - \calN_\lambda(k t_0)\right] - \sum_{j = 1}^h \left[\calA_k(t, j) - \calA_k(t_0, j)\right] \quad \text{for all} \quad t \geq 0.
		\end{align*}
		This process counts the number of tasks that arrive when all server pools have at least $h$ tasks, but relative to the number of tasks that arrive over the interval $[0, t_0]$ to server pools witht at least $h$ tasks. It follows from \eqref{ch2-seq: law of large numbers for arrivals} that
		\begin{align}
			\label{ch2-eq: fluid limit of auxiliary counting process}
			\lim_{k \to \infty} \sup_{t \in [0, T]} \left|\calC_k(t) - \left[\lambda (t - t_0) - \sum_{j = 1}^h \left(\ba(t, j) - \ba(t_0, j)\right)\right]\right| = 0
		\end{align}
		for all $T \geq 0$. Moreover, if we let $\set{\theta_{kj}}{j \geq 1}$ be the jump times of $\calC_k$ then
		\begin{align*}
			\calA_k(t, i) - \calA_k(t_0, i) = \frac{1}{k}\sum_{j = 1}^{\calC_k(t)} \ind{U_j \in I_i\left(\bq_k\left(\theta_{kj}^-\right)\right)} \quad \text{for all} \quad i \geq h + 1.
		\end{align*}
		
		Reasoning as in the proof of (a) we see that
		\begin{align*}
			\dot{\ba}(t_0, i) = \left[\lambda - \sum_{j = 1}^h \dot{\ba}(t_0, j)\right]\left[\bq(t_0, i - 1) - \bq(t_0, i)\right] \quad \text{for all} \quad i \geq h + 1.
		\end{align*}
		This equation is analogous to \eqref{ch2-eq: computing the derivative of a} after taking the limits with respect to $k$, then $t$ and finally $\varepsilon$; the limit with respect to $k$ follows from \eqref{ch2-eq: fluid limit of auxiliary counting process}. Using \eqref{ch2-eq: derivatives for i <= h when qh = 1}, we get
		\begin{align*}
			\dot{\ba}(t_0, i) = \left[\lambda - h\left(1 - \bq(t_0, h + 1)\right)\right]\left[\bq(t_0, i - 1) - \bq(t_0, i)\right] \quad \text{for all} \quad i \geq h + 1,
		\end{align*}
		which completes the proof of (c).
	\end{proof}
	
	\begin{remark}
		\label{ch2-rem: fluid limit when the threshold is locally constant}
		The lemma holds if the threshold is constant only in a neighborhood of the regular point $t_0$. It suffices that there exists $\varepsilon > 0$ such that $\ell_k(t) = \ell$ for all $t \in (t_0 - \varepsilon, t_0 + \varepsilon)$ and all large enough $k \in \calK$. The proof only changes slightly.
	\end{remark}
	
	We can now complete the proof of Theorem \ref{ch2-the: static fluid limit}.
	
	\begin{proof}[Proof of Theorem \ref{ch2-the: static fluid limit}]
		It follows from Theorem \ref{ch2-the: tightness} and Lemma \ref{ch2-lem: derivatives} that there exists a set of probability one $\Gamma_\infty$ with the following property. If $\omega \in \Gamma_\infty$, then every subsequence of $\set{\bq_n(\omega)}{n \geq 1}$ has a further subsequence that converges uniformly over compact sets to a function $\bq(\omega) \in D_{\R^\N}[0, \infty)$ that satisfies \eqref{ch2-eq: static fluid limit}.
		
		It remains to prove that there exists a subset of probability one of $\Gamma_\infty$ such that the subsequential limits $\bq(\omega)$ take values in $Q$ for all $\omega$ in this set. Specifically,
		\begin{align}
			&0 \leq \bq(\omega, t, i + 1) \leq \bq(\omega, t, i) \leq \bq(\omega, t, 0) = 1 \quad \text{for all} \quad t \geq 0 \quad \text{and} \quad i \geq 1 , \label{ch2-aux: natural inequalities} \\
			&\sum_{i = 1}^\infty \bq(\omega, t, i) < \infty \quad \text{for all} \quad t \geq 0. \label{ch2-aux: finite sum}
		\end{align}
		
		If $\set{\bq_k(\omega)}{k \in \calK}$ converges to $\bq(\omega)$ in $D_{\R^\N}[0, \infty)$, then
		\begin{equation*}
			\bq(\omega, t, i) = \lim_{k \to \infty} \bq_k(\omega, t, i) \quad \text{for all} \quad t \geq 0 \quad \text{and} \quad i \geq 1.
		\end{equation*}
		Hence, \eqref{ch2-aux: natural inequalities} follows from the fact that $\bq_k(\omega, t) \in Q$ for all $k \in \calK$ and $t \geq 0$. In order to obtain \eqref{ch2-aux: finite sum}, we will construct a set of probability one where the number of arrivals on any given finite interval of time is suitably bounded, and we will use the fact that \eqref{ch2-aux: finite sum} holds for the initial occupancy state $q_0$. 
		
		Specifically, let $\theta_t \defeq \lambda t (\e - 1) + 1$ for all $t \geq 0$ and note that
		\begin{equation*}
			\prob\left(\calN_\lambda(n t) > n\theta_t\right) \leq \frac{\e^{\lambda t(\e - 1) n}}{\e^{\theta_tn}} = \e^{-n},
		\end{equation*}
		by a Chernoff bound. It follows from the Borel-Cantelli lemma that
		\begin{equation*}
			E_t \defeq \set{\omega \in \Omega}{\calN_\lambda(n t) \leq n\theta_t\ \text{for all large enough}\ n}
		\end{equation*}
		has probability one. Therefore, the set
		\begin{equation*}
			\Gamma \defeq \bigcap_{t \in \N} \left(E_t \cap \Gamma_\infty\right)
		\end{equation*}
		has probability one. In particular, for each $\omega \in \Gamma$ and $t \geq 0$, we have
		\begin{equation*}
			\sum_{i = 1}^\infty \calA_n(\omega, t, i) \leq \frac{\calN_\lambda(nt)}{n} \leq \theta_t \quad \text{for all large enough} \quad n.
		\end{equation*}
		
		Fix $\omega \in \Gamma$ and suppose $\set{\bq_k(\omega)}{k \in \calK}$ converges to $\bq(\omega)$ in $D_{\R^\N}[0, \infty)$. Then
		\begin{align*}
			\sum_{i = 1}^m \bq(\omega, t, i) &= \lim_{k \to \infty} \sum_{i = 1}^m \left[\bq_k(\omega, 0, i) + \calA_k(\omega, t, i) - \calD_k(\omega, t, i)\right] \\
			&\leq \sum_{i = 1}^m q_0(\omega, i) + \limsup_{k \to \infty} \sum_{i = 1}^m \calA_k(\omega, t, i) \leq \sum_{i = 1}^m q_0(\omega, i) + \theta_t \quad \text{for} \quad m \geq 1.
		\end{align*}
		Since $m$ is arbitrary and $q_0(\omega) \in Q$, we conclude that \eqref{ch2-aux: finite sum} holds for $\omega \in \Gamma$.
	\end{proof}
	
	\subsection{Systems with a dynamic threshold}
	\label{ch2-sap: linear schemes}
	
	In this section we consider systems where the threshold is adjusted over time using the control rule of Section \ref{ch2-sec: learning the optimal threshold}, and we provide the proof of Theorem \ref{ch2-the: dynamic fluid limit}. As in the statement of the theorem, we assume that the initial occupancy states $\bq_n(0)$ converge almost surely to a random variable $q_0$, as in \eqref{ch2-ass: convergence of initial conditions}. Similarly, we assume that there exists a random variable $\ell_0$ with values in $\N$ such that
	\begin{align}
		\label{ch2-ass: convergence of initial threshold}
		\lim_{n \to \infty} \ell_n(0) = \ell_0 \quad \text{almost surely}.
	\end{align}
	We further assume that the initial condition $\left(q_0, \ell_0\right)$ satisfies
	\begin{align}
		\label{ch2-ass: nice initial condition}
		q_0\left(\ell_0\right) > \alpha \quad \text{and} \quad q_0\left(h_0\right) < 1 \quad \text{almost surely},
	\end{align}
	where $h_0 \defeq \ell_0 + 1$. As noted right before Theorem \ref{ch2-the: dynamic fluid limit}, this technical assumption ensures that the limiting threshold will not be modified at time zero.
	
	In order to prove the theorem, we approximate the systems $\bs_n = (\bq_n, \ell_n)$ by systems where the threshold is updated only finitely many times. Specifically, we let $\bs_n^m = (\bq_n^m, \ell_n^m)$ denote a system with $n$ server pools where the threshold is updated only the first $m$ times that the update condtions described in Section \ref{ch2-sub: learning rule} are met. Note that $\bs_n^m$ behaves exactly as $\bs_n$ until the $(m + 1)$th threshold update.
	
	Formally, we define $\tau_n^0 \defeq 0$ and we let $\set{\tau_n^m}{m \geq 1}$ denote the times at which the threshold changes in $\bs_n$. The process $\bs_n^m$ may be constructed on $(\Omega, \calF, \prob)$ as in Section \ref{ch2-sap: construction on a common probability space}, from the stochastic equations \eqref{ch2-eq: implicit equations} with \eqref{ch2-seq: threshold} replaced by
	\begin{align}
		\label{ch2-seq: threshold revisited}
		\ell(t) = \ell_n(t)\ind{t < \tau_n^m} + \ell_n(\tau_n^m)\ind{t \geq \tau_n^m}.
	\end{align}
	As in Section \ref{ch2-sap: construction on a common probability space}, the stochastic equations \eqref{ch2-seq: aggregate state}-\eqref{ch2-seq: threshold revisited} have a unique solution $\bs_n^m = (\bq_n^m, \ell_n^m)$ defined on $[0, \infty)$ with probability one. Furthermore, we have
	\begin{align}
		\label{ch2-eq: approximation by finite updates systems}
		\bq_n^m(t) = \bq_n(t) \quad \text{for all} \quad t \in \left[0, \tau_n^{m + 1}\right).
	\end{align}
	
	\begin{proposition}
		\label{ch2-prop: simultaneous relative compactness}
		There exists a set of probability one $\Gamma$ where the assumptions \eqref{ch2-ass: convergence of initial conditions}, \eqref{ch2-ass: convergence of initial threshold} and \eqref{ch2-ass: nice initial condition} hold. In addition, $\set{\bq_n^m(\omega)}{n \geq 1}$ and $\set{\bq_n(\omega)}{n \geq 1}$ are relatively compact in $D_{\R^\N}[0, \infty)$ for all $m \geq 1$ and $\omega \in \Gamma$, and the limit of every subsequential limit is a function with values in $Q$ and locally Lipschtiz components.
	\end{proposition}
	
	\begin{proof}
		Given any $m \geq 0$, the almost sure relative compactness of $\set{\bq_n^m}{n \geq 1}$ can be established as in Theorem \ref{ch2-the: tightness}. The arguments in the proof of Theorem \ref{ch2-the: static fluid limit} imply that every subsequential limit is a function with values in $Q$. These remarks also apply to $\set{\bq_n}{n \geq 1}$, so $\Gamma$ is the intersection of countably many sets of probability one.
	\end{proof}
	
	In order to prove Theorem \ref{ch2-the: dynamic fluid limit}, we will resort to an inductive argument using the systems $\bs_n^m$. The inductive step will be carried out in Lemma \ref{ch2-lem: induction step}, which depends on the following property of the differential equation \eqref{ch2-eq: static fluid limit}.
	
	\begin{lemma}
		\label{ch2-lem: time between consecutive threshold decreases is positive}
		Consider a solution $\map{\bq}{[a, b]}{Q}$ of \eqref{ch2-eq: static fluid limit} for a threshold $\ell \in \N$, and suppose that $\bq(t, \ell) \geq \alpha$ for all $t \in [a, b)$ and $\bq(b, \ell) = \alpha$. Then $\bq(b, \ell - 1) > \alpha$.
	\end{lemma}
	
	\begin{proof}
		Suppose that the claim is false. This implies that there exists $1 \leq i \leq \ell - 1$ such that $\bq(b, \ell - i) = \alpha$. Let $m$ be the largest integer $i$ with this property and let
		\begin{align*}
			\boldf(t) \defeq \lambda p_{\ell - m}(\bq(t), \ell) - (\ell - m)\left[\bq(t, \ell - m) - \bq(t, \ell - m + 1)\right] \quad \text{for all} \quad t \in [a, b].
		\end{align*}
		Note that $\dot{\bq}(\ell - m) = \boldf$ almost everywhere on $[a, b]$ and $\boldf$ is continuous in a left neighborhood of $b$, because $\bq$ has continuous coordinate functions and $\bq(b, \ell) < 1$. By $\alpha = \bq(b, \ell - m) \geq \bq(b, \ell - m + 1) \geq \bq(b, \ell) = \alpha$ and the maximality of $m$, we have
		\begin{align*}
			\boldf(b) = \lambda\left[\frac{\bq(b, \ell - m - 1) - \bq(b, \ell - m)}{1 - \bq(b, \ell)}\right] > 0.
		\end{align*}
		By continuity, there exists an open left neighborhood of $b$ where $\boldf \geq \boldf(b) / 2$. Since  $\bq(t, \ell - m) \geq \bq(t, \ell) \geq \alpha$ for all $t \in [a, b)$ and $\dot{\bq}(t, \ell - m) \geq \boldf(b) / 2 > 0$ almost everywhere on a left neighborhood of $b$, we have a contradiction: $\bq(b, \ell - m) > \alpha$.
	\end{proof}
	
	\begin{lemma}
		\label{ch2-lem: induction step}
		We fix an arbitrary $\omega \in \Gamma$ which is omitted from the notation for brevity. Suppose that there exist $m \geq 0$, a strictly increasing sequence $\set{\tau_i \geq 0}{0 \leq i \leq m}$ and a sequence of threshold values $\set{l_i \in \N}{0 \leq i \leq m}$ such that
		\begin{align*}
			\lim_{k \to \infty} \tau_k^i = \tau_i \quad \text{and} \quad \lim_{k \to \infty} \ell_k\left(\tau_k^i\right) = l_i \quad \text{for all} \quad 0 \leq i \leq m,
		\end{align*}
		where $k$ takes values in a sequence $\calK \subset \N$ such that $\set{\bq_k^m}{k \in \calK}$ and $\set{\bq_k}{k \in \calK}$ converge in $D_{\R^\N}[0, \infty)$ to functions $\bq^m$ and $\bq$, respectively. Then
		\begin{align}
			\label{ch2-aux: fluid dynamics}
			\dot{\bq}^m(i) = \lambda p_i(\bq^m, l_j) - i\left[\bq^m(i) - \bq^m(i + 1)\right] \quad \text{for all} \quad i \geq 1
		\end{align}
		almost everywhere on $\left[\tau_j, \tau_{j + 1}\right)$ if $0 \leq j < m$ and $\left[\tau_m, \infty\right)$ if $j = m$. Assume that
		\begin{align}
			\label{ch2-eq: good threshold assumption}
			\bq(\tau_m, l_m) > \alpha \quad \text{and} \quad \bq(\tau_m, h_m) < 1, \quad \text{where} \quad h_m \defeq l_m + 1.
		\end{align}
		If we let $\tau_{m + 1} \defeq \inf \set{t \geq 0}{\bq^m(t) \neq \bq(t)} \in (0, \infty]$, then
		\begin{equation}
			\label{ch2-aux: limit of update times}
			\lim_{k \to \infty} \tau_k^{m + 1} = \tau_{m + 1} > \tau_m.
		\end{equation}
		When $\tau_{m + 1} < \infty$, we have $\bq\left(\tau_{m + 1}, l_m\right) = \alpha$ or $\bq\left(\tau_{m + 1}, h_m\right) = 1$, and
		\begin{align}
			\label{ch2-eq: ell m + 1}
			l_{m + 1} \defeq \lim_{k \to \infty} \ell_k\left(\tau_k^{m + 1}\right) = \begin{cases}
				l_m - 1 & \text{if} \quad \bq(\tau_{m + 1}, l_m) = \alpha, \\
				l_m + 1 & \text{if} \quad \bq(\tau_{m + 1}, h_m) = 1.
			\end{cases}
		\end{align}
		If $\tau_{m + 1} < \infty$, then we further have:
		\begin{equation}
			\label{ch2-eq: last part of c}
			\bq\left(\tau_{m + 1}, l_{m + 1}\right) > \alpha \quad \text{and} \quad \bq\left(\tau_{m + 1}, h_{m + 1}\right) < 1.
		\end{equation}
	\end{lemma}
	
	\begin{proof}
		Let us prove \eqref{ch2-aux: fluid dynamics} only for $j < m$ since the same arguments apply for $j = m$. For each $t \in \left(\tau_j, \tau_{j + 1}\right)$, we have $\ell_k^m(s) = l_j$ for all $s \in (t - \varepsilon, t + \varepsilon)$, some $\varepsilon > 0$ and all sufficiently large $k \in \calK$. As a result, it follows from Lemma \ref{ch2-lem: derivatives} and Remark \ref{ch2-rem: fluid limit when the threshold is locally constant} that \eqref{ch2-aux: fluid dynamics} holds almost everywhere on $(t - \varepsilon, t + \varepsilon)$. Since $t$ is arbitrary, we conclude that the differential equations hold almost everywhere on $\left[\tau_j, \tau_{j + 1}\right)$.
		
		Since $\bq(l_m)$ and $\bq(h_m)$ are continuous functions by Proposition \ref{ch2-prop: simultaneous relative compactness}, we conclude from the uniform convergence of $\bq_k$ to $\bq$ and from \eqref{ch2-eq: good threshold assumption} that
		\begin{equation}
			\label{ch2-aux: conditions at tau m}
			\bq_k(t, l_m) \geq \alpha + \varepsilon \quad \text{and} \quad \bq_k(t, h_m) \leq 1 - \varepsilon \quad \text{for all} \quad t \in \left(\tau_m - \varepsilon, \tau_m + \varepsilon\right)
		\end{equation}
		for some $\varepsilon > 0$ and all sufficiently large $k \in \calK$. It follows that
		\begin{equation*}
			\tau \defeq \liminf_{k \to \infty} \tau_k^{m + 1} \geq \tau_m + \varepsilon,
		\end{equation*}
		because $\tau_k^m \in \left(\tau_m - \varepsilon, \tau_m + \varepsilon\right)$ for all large enough $k \in \calK$ and \eqref{ch2-aux: conditions at tau m} implies that the next threshold update cannot occur until $\tau_m + \varepsilon$. This proves the strict inequality on the right-hand side of \eqref{ch2-aux: limit of update times}. Furthermore, the following two properties hold:
		\begin{enumerate}
			\item[(i)] $\bq(t, l_m) \geq \alpha$ for all $t \in [\tau_m, \tau]$,
			
			\item[(ii)] $\bq(h_m + 1)$ is nonincreasing in $\left[\tau_m, \tau\right]$.
		\end{enumerate}
		
		For (i), observe that $\bq_k\left(\ell_k\right) \geq \alpha$ at every arrival time in $\left(\tau_k^m, \tau_k^{m + 1}\right)$ since the threshold does not change in this interval. Any time $t \in (\tau_m, \tau)$ can be approached by a sequence of such arrival times, indexed by $k \in \calK$, and we have $\ell_k(t) = l_m$  for all large enough $k$. It follows that (i) holds in $(\tau_m, \tau)$ and the continuity of $\bq(l_m)$ implies that (i) also holds in the closure. For (ii), note that $\bq_k\left(h_k\right) < 1$ at $\tau_k^m$ for all sufficiently large $k \in \calK$ by \eqref{ch2-aux: conditions at tau m}. Therefore, $\bq_k\left(h_k\right) < 1$ in $\left[\tau_k^m, \tau_k^{m + 1}\right)$ since an arrival that would make $\bq_k\left(h_k\right)$ reach one would also lead to a threshold increase. We conclude that all tasks are dispatched to server pools with at most $\ell_k$ tasks and thus $\bq_k\left(h_k + 1\right)$ is nonincreasing. Given $\tau_m < c < d < \tau$, we have $h_k(t) = h_m$ for all sufficiently large $k \in \calK$ and every $t \in (c, d)$. This implies that (ii) holds on any such interval $(c, d) \subset \left[\tau_m, \tau\right]$, and thus also in the interval $[\tau_m, \tau]$.
		
		In order to prove the equality in \eqref{ch2-aux: limit of update times}, we first establish that $\tau \geq \tau_{m + 1}$. Next we assume that the latter inequality does not hold and we arrive to a contradiction. If $\tau < \tau_{m + 1}$, then there exists $\calJ \subset \calK$ such that $\tau_j^{m + 1} \to \tau < \tau_{m + 1}$ as $j \to \infty$. An update occurs at $\tau_j^m$ and $\ell_j\left(\tau_j^m\right) = l_m$ for all large enough $j \in \calJ$, thus
		\begin{equation}
			\label{ch2-aux: update conditions}
			\bq(\tau, l_m) = \alpha \quad \text{or} \quad \bq(\tau, h_m) = 1.
		\end{equation}
		
		Suppose that $\bq(\tau, l_m) = \alpha$. It follows from \eqref{ch2-eq: approximation by finite updates systems} that (i) holds for $\bq^m$. Applying Lemma \ref{ch2-lem: induction step} in the interval $[\tau_m, \tau]$ to $\bq^m$, we conclude that
		\begin{equation}
			\label{ch2-aux: no update 1}
			\bq(\tau, l_m - 1) = \bq^m(\tau, l_m - 1) > \alpha \quad \text{and} \quad \bq(\tau, l_m) = \alpha < 1.
		\end{equation}
		Arguing as when we proved the strict inequality in \eqref{ch2-aux: limit of update times}, we obtain
		\begin{equation*}
			\zeta \defeq \liminf_{j \to \infty} \tau_j^{m + 2} > \tau.
		\end{equation*}
		Using Lemma \ref{ch2-lem: derivatives} and Remark \ref{ch2-rem: fluid limit when the threshold is locally constant}, we conclude that $\bq$ satisfies \eqref{ch2-aux: fluid dynamics} with $l_j = l_m - 1$ almost everywhere on $[\tau, \zeta)$. Because \eqref{ch2-aux: fluid dynamics} holds almost everywhere on $[\tau, \zeta)$ for $\bq^m$, but with $l_j = l_m$, we arrive to the following contradiction: $\bq$ and $\bq^m$ cannot coincide in any right-neighborhood of $\tau < \tau_{m + 1}$.
		
		Suppose that $\bq(\tau, h_m) = 1$. Then we conclude from (ii) that
		\begin{equation}
			\label{ch2-aux: no update 2}
			\bq(\tau, h_m) = 1 > \alpha \quad \text{and} \quad \bq(\tau, h_m + 1) \leq \bq(\tau_m, h_m + 1) \leq \bq(\tau_m, h_m) < 1.
		\end{equation}
		Arguing as in the previous case, where $\bq(\tau, l_m) = \alpha$, we reach the same contradiction.
		
		We have proved that $\tau \geq \tau_{m + 1}$, which completes the proof of \eqref{ch2-aux: limit of update times}, and also of the entire lemma in the case where $\tau_{m + 1} = \infty$. Hence, we assume in the sequel that $\tau_{m + 1} < \infty$. To complete the proof of \eqref{ch2-aux: limit of update times}, we establish that
		\begin{equation*}
			\limsup_{k \to \infty} \tau_k^{m + 1} \leq \tau_{m + 1}.
		\end{equation*}
		If this did not hold, then it would be possible to find $\varepsilon > 0$ and a subsequence $\calJ \subset \calK$ such that $\tau_j^{m + 1} \geq \tau_{m + 1} + \varepsilon$ for all $j \in \calJ$. But this implies that
		\begin{equation*}
			\bq^m(t) = \lim_{j \to \infty} \bq_j^m(t) = \lim_{j \to \infty} \bq_j(t) = \bq(t) \quad \text{for all} \quad t \in [0, \tau_{m + 1} + \varepsilon),
		\end{equation*}
		which contradicts the definition of $\tau_{m + 1}$. Thus, \eqref{ch2-aux: limit of update times} holds.
		
		The fact that $\bq\left(\tau_{m + 1}, l_m\right) = \alpha$ or $\bq\left(\tau_{m + 1}, h_m\right) = 1$ can be established in the same way as \eqref{ch2-aux: update conditions}, and \eqref{ch2-eq: ell m + 1} is a direct consequence of how the threshold is updated. Also, \eqref{ch2-eq: last part of c} is proved in the same way as \eqref{ch2-aux: no update 1} and \eqref{ch2-aux: no update 2}.
	\end{proof}
	
	We are now ready to prove Theorem \ref{ch2-the: dynamic fluid limit}.
	
	\begin{proof}[Proof of Theorem \ref{ch2-the: dynamic fluid limit}.]
		We fix $\omega \in \Gamma$ and we omit it from the notation for brevity. Using Proposition \ref{ch2-prop: simultaneous relative compactness} and a diagonal argument, we conclude that every sequence of natural numbers has a subsequence $\calK$ such that the sequences $\set{\bq_k^m}{k \in \calK}$ with $m \geq 0$ and $\set{\bq_k}{k \in \calK}$ converge uniformly over compact sets to functions with values in $Q$ and locally Lipschitz coordinate functions. By \eqref{ch2-ass: convergence of initial threshold} and \eqref{ch2-ass: nice initial condition}, Lemma \ref{ch2-lem: induction step} holds with $m = 0$. Applying it recursively, we obtain a sequence of strictly increasing times $\set{\tau_j \geq 0}{0 \leq j < \eta}$ and threshold values $\set{l_j \in \N}{0 \leq j < \eta}$, with a possibly infinite $\eta$, such that (a) and (c) of Definition \ref{ch2-def: fluid system} hold for the limit $\bq$ of $\bq_k$.
		
		Next we prove that (b) of Definition \ref{ch2-def: fluid system} also holds. This implies that the latter sequences and $\bq$ form a fluid system. The fluid system has $\tau_\eta = \infty$ by Proposition~\ref{ch2-prop: infinitely many updates cannot occur in finite time}. Also, Lemma \ref{ch2-lem: induction step} implies that the threshold processes $\set{\ell_k}{k \in \calK}$ converge in $S[0, \infty)$ to the threshold of the fluid system, which is defined as
		\begin{equation*}
			\ell(t) \defeq l_j \quad \text{for all} \quad t \in \left[\tau_j, \tau_{j + 1}\right) \quad \text{and} \quad 0 \leq j < \eta.
		\end{equation*}
		
		The fact that $\bq\left(t, \ell(t)\right) \geq \alpha$ for all $t \in [0, \tau_\eta)$ can be established in the same way as we proved property (i) in the proof of Lemma \ref{ch2-lem: induction step}. As in the proof of (ii), we conclude that all tasks are sent to server pools with at most $l_j$ tasks in the interval $(\tau_k^j, \tau_k^{j + 1})$ for all large enough $k \in \calK$; where $\tau_k^{j + 1} = \infty$ if $\eta < \infty$ and $j = \eta- 1$. As in the proof of Lemma \ref{ch2-lem: derivatives}, we may show that $p_i\left(\bq, l_j\right) = 0$ for all $i \geq h_j + 1$ almost everywhere on the interval $\left(\tau_j, \tau_{j + 1}\right)$. This implies that $\bq\left(t, h(t)\right) < 1$ almost everywhere on $[0, \tau_\eta)$, and hence property (b) of Definition \ref{ch2-def: fluid system} holds.
	\end{proof}
	
\end{appendices}
	
\newcommand{\noop}[1]{}
\bibliographystyle{IEEEtranS}
\bibliography{IEEEabrv,bibliography}
	
\end{document}